\documentclass[draftclsnofoot, onecolumn]{IEEEtran}

\usepackage{amsmath,amssymb,mathtools,mathrsfs,graphicx,overpic,color,verbatim,stackrel,url,float,enumerate,cleveref}

\newcommand{\tA}{\tilde{A}}
\newcommand{\tB}{\tilde{B}}

\newcommand{\px}{P_X}
\newcommand{\py}{P_Y}

\newcommand{\alphabet}[1] {{\mathsf #1}}

\newcommand{\pygx}{P_{Y|X}}

\newcommand{\OTTcost}{\mathcal{L}}

\newtheorem{theorem}{Theorem}[section]
\newtheorem{lemma}[theorem]{Lemma}
\newtheorem{definition}[theorem]{Definition}
\newtheorem{corollary}[theorem]{Corollary}
\newtheorem{remark}[theorem]{Remark}
\newtheorem{example}[theorem]{Example}
\newtheorem{assumption}[theorem]{Assumption}

\newcommand{\BlackBox}{\rule{1.5ex}{1.5ex}}  %
\newenvironment{proof}{\par\noindent{\bf Proof:\
}}{\hfill\BlackBox\\[2mm]}

\newcommand{\encpolicy}{\Pi}

\newcommand{\bnu}{\bar{\nu}}

\newcommand{\tail}{\mathcal{T}}

\newcommand{\cU}{\alphabet{U}}

\newcommand{\parenth}[1] {\left(#1\right)}
\newcommand{\abs}[1] {\left|#1\right|}
\newcommand{\braces}[1] {\left\{#1\right\}}
\newcommand{\brackets}[1] {\left[#1\right]}

\newcommand{\reals}{\mathbb{R}}
\newcommand{\R}{\reals}

\newcommand{\indicatorvbl}[1] {1_{\braces{#1}}}

\newcommand{\deq}{\equiv}

\renewcommand{\P}{\mathbb{P}}

\newcommand{\E}{ {\mathbb E}}
\newcommand{\prob}[1] {\P\parenth{#1}}

\newcommand{\probSimplex}[1]{ \mathcal{P}\parenth{\alphabet{#1}}}

\newcommand{\cF}{ {\cal F}}
\newcommand{\cG}{ {\cal G}}
\newcommand{\cH}{ {\cal H}}
\newcommand{\cB} {{\cal B}} %

\newcommand{\cV}{ {\alphabet{V}}}

\newcommand{\kldist}[2] {D \parenth{#1\|#2}}

\newcommand{\capacityCostFnVal}[1]{C\parenth{\statecostfn,\PDMC,#1}}
\newcommand{\capacityCostFn}{\capacityCostFnVal{\statecostval}}

\newcommand{\PDMC}{P_{\Out|\State}}

\newcommand{\norm}[1]{\left\|#1\right\|}
\newcommand{\Src}{W}

\newcommand{\cSrc}{\alphabet{\Src}}

\newcommand{\State}{X}

\newcommand{\cState}{\alphabet{\State}}

\newcommand{\cY}{\alphabet{Y}}%
\newcommand{\cX} {\alphabet{X}}%

\newcommand{\Out}{Y}

\newcommand{\bpi}{\bar{\pi}}

\newcommand{\statecostfn}{\eta}

\newcommand{\statecostval}{L}
\newcommand{\bitm}{\begin{itemize}}
\newcommand{\eitm}{\end{itemize}}
\newcommand{\benum}{\begin{enumerate}}
\newcommand{\eenum}{\end{enumerate}}
\newcommand{\beqa}{\begin{eqnarray}}
\newcommand{\eeqa}{\end{eqnarray}}
\newcommand{\beqas}{\begin{eqnarray*}}
\newcommand{\eeqas}{\end{eqnarray*}}
\newcommand{\baln}{\begin{align}}
\newcommand{\ealn}{\end{align}}
\newcommand{\balns}{\begin{align*}}
\newcommand{\ealns}{\end{align*}}

\newcommand{\cW}{ \cSrc}
\renewcommand{\P}{\mathbb{P}}
\newcommand{\bP}{\mathbb{\overline{P}}}

\newcommand{\bprob}[1] {\bP\parenth{#1}}
\newcommand{\bE}{ {\overline{\mathbb E}}}

\newcommand{\tW}{{\tilde{W}}}

\newcommand{\tw}{\tilde{w}}

\newcommand{\tu}{\tilde{u}}
\newcommand{\tv}{\tilde{v}}

\newcommand{\tS}{\tilde{S}}
\newcommand{\tP}{\tilde{\P}}

\newcommand{\cFYinf}{\cF^Y_{1:\infty}}
\newcommand{\cFYn}{\cF^Y_{1:n}}

\newcommand{\cFYnb}{\cF^Y_{1:n-1}}

\newcommand{\itt}[1]{\textit{#1}} 
\newcommand{\post}{\psi}
\renewcommand{\encpolicy}{\Psi}

\newcommand{\bpost}{\bar{\post}}

\title{Construction and Analysis of Posterior Matching in Arbitrary Dimensions via Optimal Transport}

\author{Diego A. Mesa \thanks{
D. A. Mesa is  a postdoctoral fellow in the Departments of Biomedical Informatics, and Electrical Engineering and Computer Science at Vanderbilt University and was supported by an NSF Graduate Research Fellowship. } \and Rui Ma \thanks{R. Ma
is currently a senior staff associate at Dexcomm, Inc and was supported by a 
postdoctoral fellowship from the NSF Center for Science of Information.  }  \and
Siva Gorantla \thanks{S. Gorantla is now a senior software engineer at Google, Inc.} \and  Todd P. Coleman \thanks{T. P. Coleman a Professor in the Department of Bioengineering, University of California San Diego and is supported by NSF grants CCF-1065022 and CCF-0939370, and Army Research Office grants ARO-62793-RT-REP, and ARO MURI W911NF-15-1-0479.}}

\begin{document}

\maketitle

\begin{abstract}
The posterior matching scheme, for feedback encoding of a message point lying on the unit interval over memoryless channels, maximizes mutual information for an arbitrary number of channel uses.  However, it in general does not always achieve any positive rate; so far, elaborate analyses have been required to show that it achieves any positive rate below capacity. More recent efforts have introduced a random ``dither'' shared by the encoder and decoder to the problem formulation, to simplify analyses and guarantee that the randomized scheme achieves any rate below capacity.  Motivated by applications (e.g. human-computer interfaces) where (a) common randomness shared by the encoder and decoder may not be feasible and (b) the message point lies in a higher dimensional space, we focus here on the original formulation without common randomness, and use optimal transport theory to generalize the scheme for a message point in a higher dimensional space.  By defining a stricter, almost sure, notion of message decoding, we use classical probabilistic techniques (e.g. change of measure and martingale convergence) to establish succinct necessary and sufficient conditions on when the message point can be recovered from infinite observations: Birkhoff ergodicity of a random process sequentially generated by the encoder.  We also show a surprising ``all or nothing'' result: the same ergodicity condition is necessary and sufficient to achieve any rate below capacity.  We provide applications of this message point framework in human-computer interfaces and multi-antenna communications.  %
 \end{abstract}

\begin{IEEEkeywords}
feedback communication, posterior matching, optimal transport theory, Markov chains, ergodicity, Bayesian inference
\end{IEEEkeywords}

\section{Introduction}
Consider the communication problem shown in Figure~\ref{fig:channelCodingFB}: a message $W$ is
signaled sequentially with feedback across a noisy channel.  At time step $n$,
an encoder uses the message $W$ and previous channel outputs $Y_1, \ldots,
Y_{n-1}$ to specify a signal $X_n$, which is then transmitted across the channel.

Motivated by emerging applications in human-computer interaction and the internet of things, we model the problem as $W \in \cW \subset \reals^d$ and consider optimizing over encoding strategies that map message point $W$ and previous channel outputs $Y_1, \ldots, Y_{n-1}$ into the next channel input
$X_n$.  The decoder, with knowledge of the
encoder's strategy, simply performs Bayesian updates to sequentially construct a
posterior belief $\post_n$ about the message point after observations
$Y_1,\ldots,Y_n$.  In this setup, we do not prespecify a block length; we can define \itt{reliability} in terms of
the message point $W$ being recoverable from all previous observations $Y_1, Y_2,
\ldots$.  This is equivalent to the condition that $\post_n$ tends to a point
mass at $W$.  Secondly, we can define the notion of \itt{achieving a rate} in
terms of the speed of convergence of $\post_n$ towards a point mass.

\begin{figure}[t]
    \centering
    \begin{overpic}[width=\columnwidth]{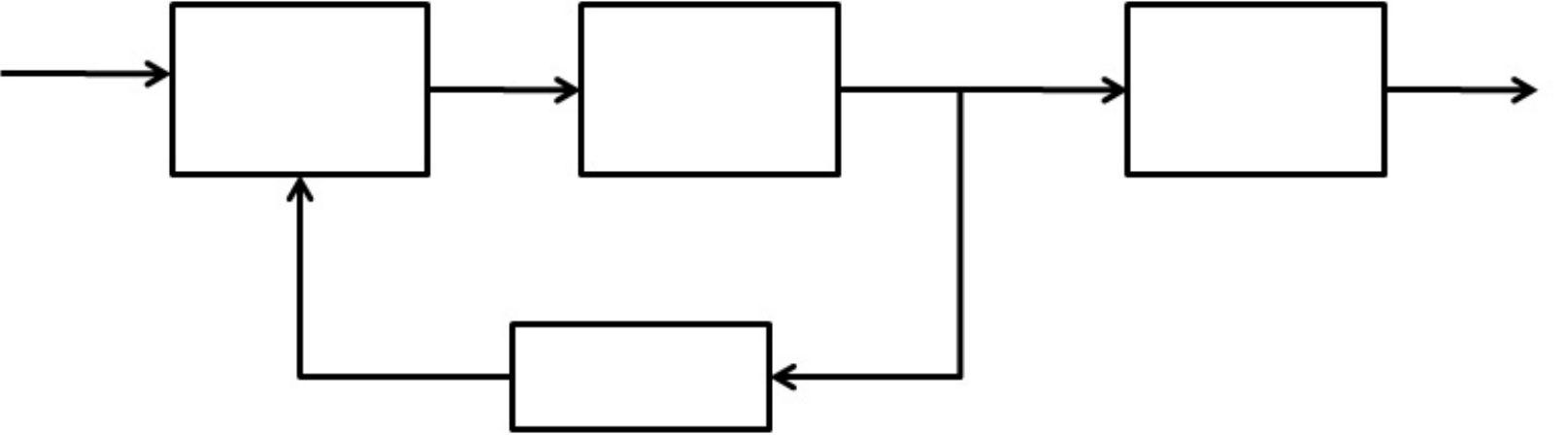}
        \put(12,23){\bf causal}
        \put(12,19){\bf encoder}
        \put(72,23){\bf posterior}
        \put(72,19){\bf update}
        \put(39,23){\bf noisy}
        \put(39,19){\bf channel}
        \put(36,3){\bf delay}
        \put(91,24){\Large $\post_{n}$}
        \put(59,24){\Large $\Out_n$}
        \put(28,24){\Large $\State_n$}
        \put(2,24){\Large $\Src$}
        \put(20,8){\Large $Y^{n-1}$}
        \put(2,8){$\reals^d$}
    \end{overpic}
    \caption{Communication of a message point $\Src$ with causal feedback over a
    memoryless channel. The message point $W$ lies in a continuum $\cW \subset
    \reals^d$ and it is our desire to optimize the encoder towards steering the
    posterior distribution on $W$ given $Y_1,\ldots,Y_n$: $\post_n$, towards a
    point mass at $W$ as rapidly as possible.}
    \label{fig:channelCodingFB}
\end{figure}
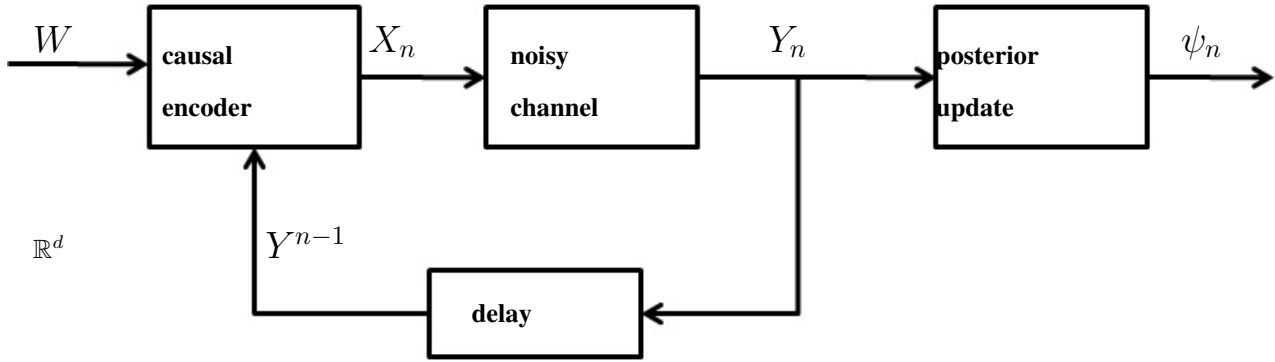

The posterior matching (PM) scheme recently developed by Shayevitz and Feder
\cite{Shayevitz2011pm} is a feedback message point encoding scheme of this
flavor when the message is a point on the unit interval, e.g. $\cW=(0,1)$.  It generalizes the scheme by Horstein that was
specific to the binary symmetric channel \cite{horstein1963stu}, and the work
of Schalkwijk and Kailath \cite{schalkwijk1966csaOne} that was specific to the
additive noise Gaussian channel.  The PM scheme has an iterative, time-invariant
state space description, and maximizes mutual information.  This is
desirable for implementation because:
\begin{itemize}
    \item The scheme does not have a pre-specified block length; it operates sequentially and at each time horizon, it maximizes the mutual information between the message point and all causal channel outputs.
    \item There is no forward error correction - it simply adapts on the fly and
    sequentially ``hands the decoder what is missing''.
    \item The scheme admits a simple time-invariant dynamical system structure.
\end{itemize}
It can be interpreted as the optimal solution to an interactive two-agent sequential decision-making 
problem consisting of a Markov source process, a causal encoder with feedback, and a 
causal decoder \cite{GorantlaColemanInformationTheoreticIT2011}.

However, from an information theoretic perspective, maximizing mutual information is necessary  but not sufficient to achieve any rate below capacity.  Indeed, examples have  been shown \cite{shayevitz2009posterior} 
for which no positive rate can be achieved for the PM scheme.  This has led to  `regularity' conditions imposed on the induced joint distribution between channel inputs and outputs to suffice for achieving any rate below capacity \cite{Shayevitz2011pm}.  It is unclear if these conditions, which conceptually guarantee that the channel law is not too sensitive to perturbations on the channel input, are strictly required to achieve any rate below capacity.
Moreover, motivated by human-computer interfaces \cite{omar2010feedback,tantiongloc2016information} and other applications, it would be desirable to have a problem formulation and solution, where the message point is in higher dimensions (e.g.  $\cW \subset \reals^d$), that maintains the same iterative-time-invariant encoding properties as well as sufficient conditions on achieving any rate below capacity.  

In their concluding remarks, Shayevitz and Feder discussed extending the
posterior matching scheme to channels with memory, and in doing so, argued that
a multidimensional message point would be required. They gave a conceptual
example based on a Markovian channel of order $d$, which would require a $d+1$
dimensional message point in order to provide the ``necessary degrees of freedom
in terms of randomness''\cite[Section VIII]{Shayevitz2011pm}. 

In what follows, we formulate a generalization to the message point feedback communication problem for higher dimensions and develop a posterior matching scheme with optimal transport theory that inherits all the desirable properties of the one-dimensional scheme. We use classical probabilistic techniques (e.g. change of measure and martingale convergence) to establish \textit{succinct} necessary and sufficient conditions on achieving capacity: Birkhoff ergodicity of a random process sequentially generated by the encoder.

\subsection{Previous and Related Work} %
\label{sub:previous_work}
Horstein first introduced a problem of this framework for the binary symmetric channel (BSC) \cite{horstein1963stu}, where the message point lies on the $(0,1)$ interval.  In this work, Horstein showed that the median of the posterior distribution is a sufficient statistic for the decoder to provide the encoder, for which the subsequent channel input signals whether the message point is larger or smaller than this threshold.  Subsequently, Schalkwijk and Kailath \cite{schalkwijk1966csaOne}\cite{schalkwijk1966csaTwo} considered signaling a message point on the $(0,1)$ interval over the additive white Gaussian channel (AWGN) with feedback.  There, they showed a close connection with estimation theory, where the minimum mean square error (MMSE) estimate of the message given all the observations plays a key role in the feedback encoder scheme.  It was also shown that not only can capacity be achieved, but also a doubly exponential error exponent.  

Shayevitz and Feder introduced the posterior matching scheme for the message point lying on the $(0,1)$ interval in \cite{Shayevitz2011pm} that is applicable to {\it any} memoryless channel.  It includes the encoding schemes for the BSC and AWGN channel by Horstein and Schalkwijk-Kailath as special cases and provides the first rigorous proof that the Horstein scheme achieves capacity on the BSC.

Applications and variations of this formulation are manyfold.  The Horstein
scheme, combined with arithmetic coding of a sequence of symbols in an ordered symbolic alphabet to represent any sequence with enumerative source encoding \cite{cover1973enumerative} as a message point on the $(0,1)$ line, has been used for brain-computer interfaces to specify a sentence or a smooth path \cite{OmarEtAkBCIIJHJCIsubmittedNov09}, to navigate mobile robots \cite{akce2013brain}, and to remotely teleoperate an unmanned aircraft  \cite{omar2010feedback}.  Many active learning problems borrow principles from posterior matching, but with  different formulations or performance criteria; generalized 20 questions for target search \cite{tsiligkaridis2014collaborative,tsiligkaridis2017decentralized} and generalized binary search for function search \cite{castro2008minimax,nowak2011geometry} serve as examples.   Naghsvar and Tavidi considered variable-length encoding with feedback and utilized principles from stochastic control and posterior matching  to demonstrate non-asymptotic upper bounds on expected code length and deterministic one-phase coding schemes that achieve capacity and attain optimal error exponents \cite{naghshvar2015extrinsic}.   More general variable-length formulations of active hypothesis testing, where the statistics of the measurement may depend on the query or are controllable, have been developed in \cite{naghshvar2013active,kaspi2018searching,chiu2016sequential,lalitha2017measurement}.  Feedback coding channels with memory using principles from posterior matching have been explored in  \cite{bae2010posterior,anastasopoulos2012sequential,wu2016zero,anastasopoulos2017variable}.

Although the PM scheme maximizes mutual information, in some
situations the posterior $\post_n$ never converges to a
point mass at $W$, implying that no positive rate is achievable. \cite[Example 11]{Shayevitz2011pm} shows that the breakdown of reliability or achieving capacity is not solely a property of the channel: for the same channel, variants of the original scheme involving measure-preserving transformations of the input can ameliorate these issues.  Sufficient conditions for
reliability \cite[Lemma 13]{Shayevitz2011pm} and achieving capacity  \cite[Theorem 4]{Shayevitz2011pm} were originally established in \cite{Shayevitz2011pm}, involving `regularity' and uniformly bounded
max-to-min ratios assumptions along with elaborate fix-point analysis.   This has led researchers to slightly alter the problem formulation and encoding schemes to more simply confirm guarantees on subsequent performance.  Li and El Gamal \cite{cheuk2015EfficientFeedback} considered a non-sequential, fixed-rate,
fixed-block-length feedback coding scheme for discrete memoryless channels (DMCs) when $\cW=(0,1)$ and introduced a random dither known to the encoder and decoder to provide a simple proof of achieving capacity.  

Also, Shayevitz and Feder  \cite{shayevitz2016simple} examined a {\it
randomized} variant of the \emph{original} PM scheme with a random dither and
were able to provide a much simpler proof of optimality over general memoryless
channels when $W \in (0,1)$. That is, as compared to the proofs above which are
restricted to non-sequential variants with a fixed number of messages and only
apply to DMCs, Shayevitz and Feder \cite{shayevitz2016simple} used a random
dither to provide a proof of optimality of the sequential, horizon-free,
\emph{randomized} posterior matching scheme over general memoryless channels.
These settings where the encoder and decoder share a common source of randomness
by way of a dither, however, may be undesirable in some situations (e.g. when
considering human involvement as described above).

Recently, we have considered the case when the message point lies in a higher dimensional space (e.g. $\cW \subset \reals^d$ for some arbitrary $d$) and constructed a feedback encoding scheme using optimal transport theory \cite{ma2011generalizing}.  
Inspired by the time-invariant dynamical systems structure of the \itt{original}
PM scheme (without shared randomness at the encoder), and motivated by communication applications where humans play a role, below we develop appropriate notions of reliability and rate achievability in
the multidimensional message point setting involving almost sure convergence.
This allows us to use classical probability tools, including change of measure and martingale convergence, that
have recently been employed by Van Handel to provide rigorous and succinct results pertaining to filter stability in hidden Markov models \cite{van2009stability,chigansky2009intrinsic}, to provide {\it succinct} necessary and
sufficient conditions for the generalized PM scheme to attain optimal
performance.  This provides a clear characterization of optimality of the original PM scheme in arbitrary dimensions in terms of Birkhoff ergodicity of an appropriate random process, without requiring the use of a
dither.

\subsection{Main Contribution}
In this paper, we address two unmet needs:
\bitm
    \item [(a)] We develop a generalization to the PM scheme for arbitrary
    memoryless channels where $\cW \subset \reals^d$ for any $d \geq 1$.
    Specifically, using optimal transport theory \cite{villani2009optimal}, we develop recursive encoding schemes that share  the same
    mutual-information maximizing and iterative, time-invariant properties; moreover, they
    reduce to that of Shayevitz and Feder \cite{Shayevitz2011pm} when
    $\cW=(0,1)$ as a special case.
    \item [(b)] We define  notions of reliability and achievability in a manner
    analogous to \cite{Shayevitz2011pm} but in terms of almost-sure convergence
    of random variables. With this, we then develop necessary and sufficient
    conditions for the scheme to be reliable and/or attain optimal convergence
    rate (e.g. achieve capacity).  We show that both of these conditions have
    the same necessary and sufficient condition: the Birkhoff ergodicity of a random
    process $(\tW_n)_{n \geq 1}$ within the encoder of a PM scheme.
\eitm
Optimal transport theory, used to construct schemes in (a), is also exploited in (b)
where an invertibility property implicit in these schemes is used to show the equivalent
conditions in (b).  The rest of the paper is organized as follows:

In Section~\ref{sec:preliminaries}, we provide definitions, terminology, and
notations that will be used throughout the paper.

In Section~\ref{sec:problemSetup}, we formulate the message point feedback
communication problem for general memoryless channels, define the performance measure of
reliability and of a rate being achievable, both in terms of almost-sure
equivalents of the notions developed by Shayevitz \& Feder
\cite{Shayevitz2011pm}.
We also provide necessary conditions for
reliability of any message feedback encoder in terms of total variation convergence between two posterior
distributions on the message point having the same observations but with
different priors on $\cW$
(Theorem~\ref{thm:reliability:totalvariationPosterior}).  We also discuss how this necessary condition and its proof are intimately related to results of Van Handel \cite{chigansky2009intrinsic} on filter stability of hidden Markov models.

In Section~\ref{sec:PMscheme}, we define posterior matching schemes in arbitrary
dimensions as time-invariant dynamical systems for an intermediate encoder state
variable $\tW_n = S_{Y_{n-1}}(\tW_{n-1})$ where the channel input satisfies $X_n
= \phi(\tW_n)$.  These schemes maximize $I(W;Y_1,\ldots,Y_n)$ for every $n \geq 1$ and satisfy a
reliability necessary condition, developed in Lemma
~\ref{lemma:reliabilityInvertibility}, that $W\equiv \tW_1$ can be recovered
from $\tW_{n+1}$ and $Y_{1:n}$ for any $n \geq 1$.  We then show in
Theorem~\ref{theorem:PM:properties} a collection of properties that PM schemes
share, including the stationary and Markov nature of the random processes
$(\tW_n)_{n \geq 1}$ and $(\tW_n,Y_n)_{n \geq 1}$. We then demonstrate in
Theorem \ref{thm:monge:phi} and Corollaries \ref{corollary:monge:unique:phi} and 
\ref{corollary:unique:S}, via the theory of optimal transport, that we can use
optimization approaches to explicitly and uniquely construct such PM schemes.

In Section~\ref{sec:reliability}, we demonstrate in
Theorem~\ref{thm:necessarySufficientConditionsPM:reliability} that the necessary
and sufficient condition for reliability of the PM scheme is the Birkhoff
ergodicity of the random process $(\tW_n)_{n \geq 1}$. This theorem uses the
unique invertibility and statistical independence properties of the PM scheme along with classical probabilistic techniques such as martingale convergence and Blackwell \& Freeman's equivalent conditions on Birkhoff ergodicity of stationary Markov chains \cite[Thm 2]{blackwell1964tail}.

In Section~\ref{sec:achievingCapacity},  we show that three things are
equivalent: (a) the PM scheme is reliable; (b) the random process $(\tW_n)_{n
\geq 1}$ is ergodic; and (c) the PM scheme achieves any rate below capacity.
This gives rise to an ``all-or-nothing'' property for PM schemes: ergodicity of
the random process $(\tW_n)_{n \geq 1}$ elucidates essentially everything about
the PM scheme's performance: either no bits can be reliably transmitted (e.g.
reliability does not hold), or any rate below capacity is achievable.
The proofs involve the notion of ``pullback intervals'' developed by Shayevitz \& Feder in \cite{Shayevitz2011pm}, the  construction of an auxiliary probability measure for which the logarithm of an appropriate Radon-Nikodym derivative is the information density, and classical probabilistic techniques including change of measure, martingale convergence, and stationarity of Markov chains.

In Section~\ref{sec:applications}, we discuss applications in brain-computer
interfaces and multi-antenna communications, and in
Section~\ref{sec:discussionConclusion} we provide some discussion and concluding
remarks.

\section{Definitions and Notations} \label{sec:preliminaries}

\subsection{Probability Definitions}
\begin{list}{\labelitemi}{\leftmargin=0.1em}
    \item Denote $a_i^j$ as $(a_i,\ldots,a_j)$ and $a^j \triangleq a_1^j$.  We
    also use $a_{i:j}$ and $a^j_i$ interchangeably.

    \item Denote $\mu$ as the Lebesgue measure on the Borel space for $\reals^d$.

    \item Denote a probability space as $(\Omega,\cF,\P)$ and the set of all
    probability measures on a measurable space $(\cX,\cF_\cX)$ as
    $\probSimplex{X}$.    %

    \item A {\em transition kernel} on $\cF_\cY \times \cX$ is a mapping
    $\pygx(\cdot|\cdot) : \cF_\cY \times \cX \to [0,1]$, such that
    $\pygx(\cdot|x) \in \probSimplex{Y}$ for every $x \in \cX$ and
    $\pygx(A|\cdot)$ is measurable for every $A \in \cF_\cY$. 
    For all $B \in
    \cF_\cY$, define the marginal distribution $\py \in \probSimplex{Y}$ by
    \begin{align*}
     \py(B) &\deq  \int_{\cX}  \pygx(B|x) \px(dx).
    \end{align*}

    \item Define the {\it join} on two $\sigma$-algebras $\mathcal{A}$ and
    $\mathcal{B}$, given by $\mathcal{A} \vee \mathcal{B}$, as the \!smallest
    $\sigma$-algebra containing both:
    \beqas
        \mathcal{A} \vee \mathcal{B} = \braces{A_i \cap B_j: A_i \in \mathcal{A}, B_j \in \mathcal{B}}
    \eeqas

    \item Denote $\sigma(Y)$ as the sigma-algebra generated by random variable $Y$ and $\sigma(Y_k,\ldots,Y_m)$ as
    \[ \sigma(Y_k,\ldots,Y_m) \triangleq \bigvee_{j=k}^{m} \sigma(Y_j).\]
    We  use $\sigma(Y_k,\ldots,Y_m)$ and $\cF^Y_{k,m}$ interchangeably.

\end{list}

\subsection{Information Theoretic Definitions}
\begin{list}{\labelitemi}{\leftmargin=0.1em}
    \item We define the {\em KL divergence} as
    \begin{align*}
        \kldist{P}{Q} \deq \begin{cases}
        \E_P[\log (dP/dQ)], & \text{if } P \ll Q\\
        + \infty, & \text{otherwise }
        \end{cases}
    \end{align*}

    \item For $P_{X,Y} \in \probSimplex{X \times \cY}$ with marginals $P_X$ and $P_Y$,  the {\it information density} is given by     \cite{verdu1994general}:
    \begin{align}
        i(x,y) %
            &\triangleq \log \frac{dP_{X,Y}}{d (P_X \times P_Y)}(x,y) \label{eqn:defn:informationdensity:b}.
    \end{align} \label{eqn:defn:informationdensity}
    The {\it mutual information} is defined as
    \begin{align}
        I(X;Y) &\triangleq \E[i(X,Y)] %
        \label{eq:defn:mutualinformation:conditionalKLdivergence:a}.
    \end{align}

\end{list}

\subsection{Definitions for Birkhoff Ergodicity of Random Processes}
\begin{list}{\labelitemi}{\leftmargin=0.1em}
    \item For two probability measures $P,Q$ defined on $\parenth{\cX,\cF}$ with
    densities with respect to the Lebesgue measure $\mu$ given by $p,q$, define
    the total variation distance $\|P-Q\|$ as:
    \begin{align}
        \|P-Q\| \triangleq& \sup_{A \in \cF} \abs{P(A)-Q(A)}  \label{eqn:defn:dTV} 
    \end{align}
    
    \item For a time-homogeneous Markov process  $(V_n)_{n \geq 1}$, denote
    $\P_\nu$ as the distribution on $(V_n)_{n \geq 1}$ with $P_{V_1}= \nu$.

    \item For a random process $(V_n)_{n \geq 1}$ on $(\Omega,\cF,\P)$, the tail
    \begin{align*}
        \tail_{V} \triangleq  \bigcap_{n \geq 1} \sigma(V_{n:\infty})
    \end{align*}
    is $\P$-trivial if $\P(E)=0$ or $\P(E)=1$ for any $E \in \tail_{V}$.

    \item A stationary Markov process $(V_n)_{n \geq 1}$ on a probability space
    $(\Omega,\cF,\P)$ is defined as {\bf $\P$-ergodic} when $\tail_{V}$ is
    $\P$-trivial.

    \item A stationary Markov process $(V_n)_{n \geq 1}$ on $(\Omega,\cF,\P)$
    with invariant measure $\nu$ on $(\cV,\cF_\cV)$ is denoted to be $\P$-mixing
    (in the ergodic theory sense) if for any $A,B \in \cF_{\mathsf{V}}$:
    \begin{equation*}
      \lim_{n \to \infty} \prob{V_n \in A, V_1 \in B} =\nu(A) \nu(B).
    \end{equation*}
    
    \item \begin{lemma}[Sec 2.5,\cite{bradley2005basic}]\label{corollary:mixingImpliesErgodicity}
    If a stationary Markov process $(V_n)_{n \geq 1}$ is $\P$-mixing then it is
    $\P$-ergodic. \end{lemma}

\end{list}

\subsection{Definitions for Optimal Transport Theory}
\begin{list}{\labelitemi}{\leftmargin=0.1em}
    \item For $u \in \reals^d$ and $M$ a $d$-by-$d$ positive definite matrix,
    denote the $M$-weighted Euclidean norm as $\|u\|_M^2 = u^T M u$.

    \item
    For $\cW \subset \reals^d$ and $\cF_{\cW}$ its Borel sets, consider two
    probability measures $P,Q \in \probSimplex{W}$. We say that a
    Borel-measurable map $S: \cW \to \cW$  {\it pushes} $P$  to $Q$, specified
    as $S \#P=Q$, if a random variable $W$ with distribution $P$ results in the 
    random variable $V \equiv S(W)$ having distribution $Q$. %

    \item 
    \begin{definition}
        A map $S: \cW \to \cW$ is a {\it diffeomorphism} if $S$ is invertible,
        $S$ is differentiable, and $S^{-1}$ is differentiable.
        \label{defn:diffeomorphism}
    \end{definition}

\end{list}
 
\section{The Message Point Feedback Communication Problem} \label{sec:problemSetup}

\newcommand{\RNderivativenu}{\frac{d\nu}{d\bnu}}
\newcommand{\RNderivativepin}{\frac{d\post^\nu_n}{d\bpost_n}(W)}
\newcommand{\RNderivativeA}{\frac{d \nu}{ d \bnu}(W)}
\newcommand{\denominatorA}{\bE \brackets{ \RNderivativeA | \cFYn}}
\newcommand{\denominatorinf}{\bE \brackets{ \RNderivativeA | \cFYinf}}

\subsection{System Description}
We consider the communication of a message point $W \in \cW$ over a memoryless
channel with causal feedback, as given in Figure~\ref{fig:channelCodingFB}.

We make the following assumption:
\begin{assumption}\label{assumption:cW}
$\cW,\cX, \cY$ each are Euclidean spaces with Borel sigma-algebras
and $\cW,\cX \subset \reals^d$.   Moreover, $\cW$ is open, bounded and convex.
\end{assumption}

Our notion of a communication system in
Figure~\ref{fig:channelCodingFB} is somewhat non-traditional because (a) we do
not a priori specify the number $n$ of channel uses and (b) by virtue of Assumption~\ref{assumption:cW}, $\cW$ is a continuous rather than discrete set of possible messages.

The encoder policy $\encpolicy=(e_n: \cW \times \cY^{n-1} \to \cX)_{n \geq 1}$
specifies the next channel input $X_n = e_n(W,Y^{n-1})$. The channel input $\State_n \in \cState$ is passed through a time-invariant memoryless channel to produce $Y_n \in \cY$:
\beqa
    \prob{\Out_n \in A | \sigma(X_{1:\infty},Y_{1:n-1}) } =
    \pygx(A|X_n). \label{eqn:defn:DMC}
\eeqa

\begin{definition}\label{defn:Pw:bnu}
Define $P_W \equiv \bnu$ as the uniform distribution on $\cW$.   We will consider two cases,  $W \sim \bnu$, and also when $W \sim \nu$, where $\nu \ll \bnu$ but is otherwise arbitrary. 

\end{definition}

Given the  distribution $\nu~(\bnu)$ on the message point $W$, the policy $\encpolicy=(e_n: \cW
\times \cY^{n-1})_{n \geq 1}$ and the channel $\pygx$, 
denoted succinctly as $\parenth{\nu,\encpolicy,\pygx}$~$\parenth{\parenth{\bnu,\encpolicy,\pygx}}$ respectively, the full distribution on  $\parenth{W, (X_n)_{n \geq1},(Y_n)_{n \geq1}}$ is specified.
For any encoder policy $\encpolicy$ and given channel $\pygx$, we will always
consider two probability spaces:
\begin{itemize}
    \item $(\Omega,\cF,\P)$ pertaining to $\parenth{\nu,\encpolicy,\pygx}$.
    \item $(\Omega,\cF,\bP)$ pertaining to
        $\parenth{\bnu,\encpolicy,\pygx}$.
\end{itemize}
Since $\cW,\cX,\cY$ are countably generated, a regular version of the posterior distribution $\post_n~(\bpost_n)$  on $W$ given $Y_{1:n}$ exists under $\P~(\bP)$, respectively.   Let $\lambda_Y$ be a $\sigma$-finite (reference)
measure on $\cY$ so that $P_{X_n,Y_n} \ll P_{X_n} \otimes \lambda_Y$, and define
$L(y|x) \triangleq \frac{dP_{Y_n|X_n=x}}{d\lambda_Y}(y)$.  Then we can describe
the recursive update to $\post_n~(\bpost_n)$, which satisfies a recursive update equation
pertaining to Bayes' rule.  Specifically, for any $A \in \cF_\cW$, we have:
\begin{subequations} \label{eqn:posteriorprobabilities}
    \begin{eqnarray}
    \post_n(A) &\triangleq& \prob{W \in A | \cFYn} \\
    \bpost_n(A) &\triangleq& \bprob{W \in A | \cFYn} %
    \end{eqnarray}
\end{subequations}
As noted in Definition~\ref{defn:Pw:bnu}, $P_W$ will always indicate $\bnu$, the uniform distribution on $\cW$ for the remainder of the manuscript, and $(\Omega,\cF,\bP)$ will be the probability space for which all formal performance criteria will be evaluated.  $\nu \ll \bnu$ and the associated probability space $(\Omega,\cF,\P)$ will be defined in a context-specific manner for certain theorems (Theorem~\ref{thm:reliability:totalvariationPosterior}, Lemma~\ref{lemma:equivalentConditions:MarkovStationaryErgodic}) and their proofs (Theorem~\ref{thm:necessarySufficientConditionsPM:reliability}).

\begin{remark}
    Our rationale for making the definitions of $\bnu$ and $\nu$ along with $\P$ and $\bP$ will become
    clearer in Section~\ref{sec:reliability}: it allows us to be notationally
    consistent with definitions used in the analysis of nonlinear filter
    stability, which explores when the posterior distribution in a hidden Markov
    model is sensitive to the initial distribution on the state variable (e.g.
    $\nu$ vs $\bnu$) \cite{van2009stability}.  One key difference, however, is
    that the nonlinear filter pertains to recursive updates of the posterior
    distribution of the \itt{current} latent state of a hidden Markov model.  In our
    setting, our focus is squarely on recursively updating $\post_n$ and
    $\bpost_n$, the posterior distribution on the latent \itt{initial} condition $W$.
\end{remark}

\subsection{Reliable Communication} \label{sec:reliableCommunication}
We now provide a formalism of achievability that implies the standard Shannon
theoretic notions in \cite{CoverThomas06}.

\begin{definition}[Reliability]\label{defn:reliability}
    An encoder policy $\encpolicy$ is reliable if $W$ is  $\cFYinf$-measurable $\bP-a.s.$
\end{definition}
In other words, reliability means that from all the observations $(Y_n)_{n \geq 1}$, one can perfectly reconstruct $W$.  This implies that any fixed number of bits can be decoded from all the measurements.  Because the number of decodable messages does not grow with the number of channel usages $n$, we can analogously say that this means achieving ``zero rate''. Figure~\ref{fig:pi} shows the posterior distribution trajectory of two encoding schemes, one that is reliable and another that is not.

\begin{figure}[t]
    \centering
    \begin{overpic}[width=\columnwidth]{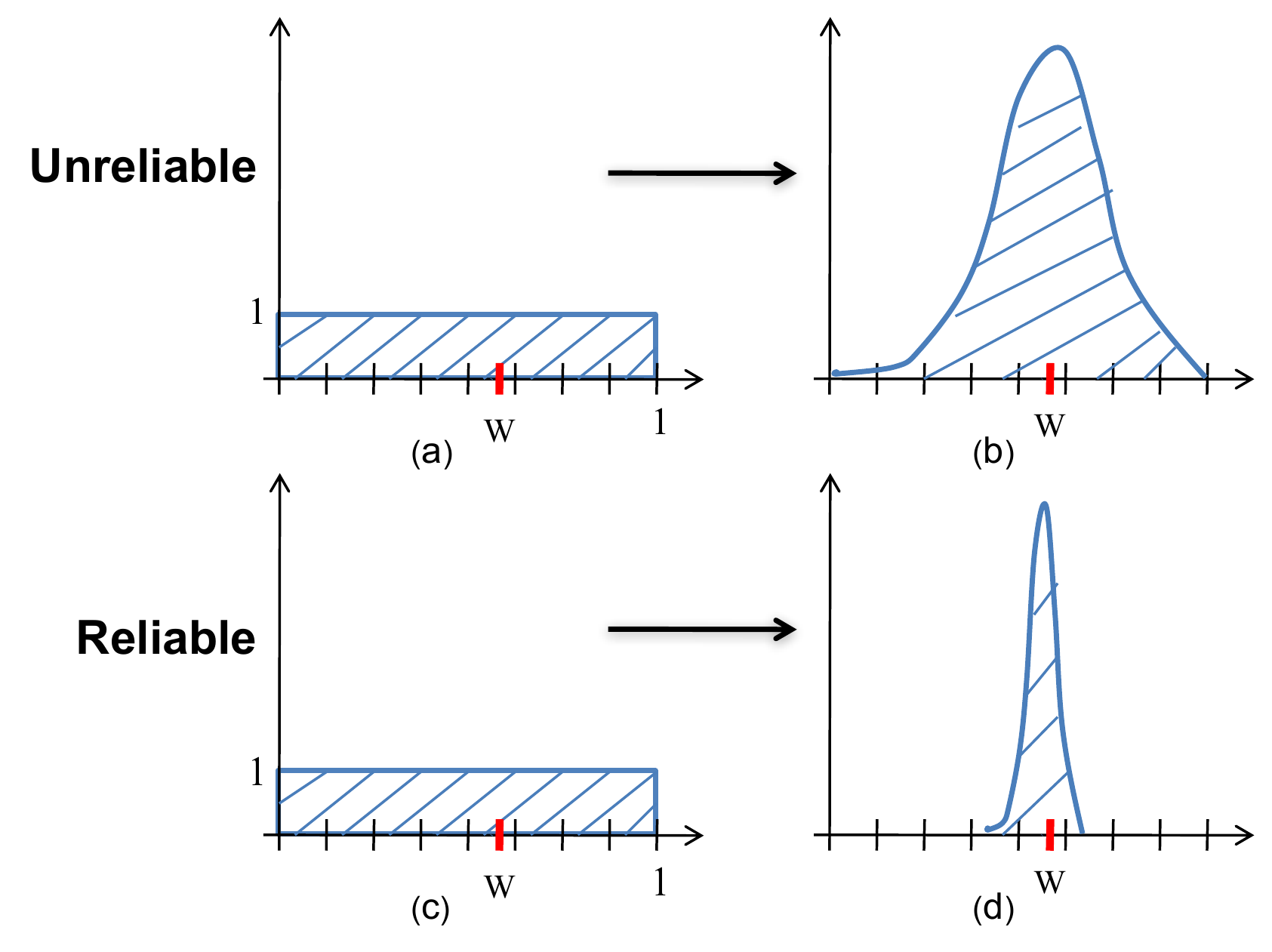}
    \put(30,55){ $\bpost_0(\cdot)$}     %
    \put(30,20){ $\bpost_0(\cdot)$}     %
    \put(85,30){ $\bpost_n(\cdot)$}     %
    \put(85,70){ $\bpost_n(\cdot)$}     %
    \end{overpic}
    \caption{(a) represents prior $\bpost_0$, (b) represents the posterior
    $\bpost_n$ after observations $y^n$ when the scheme is unreliable. (c) and
    (d) show the same as (a) and (b) respectively, except the scheme is
    reliable.}
    \label{fig:pi}
\end{figure}
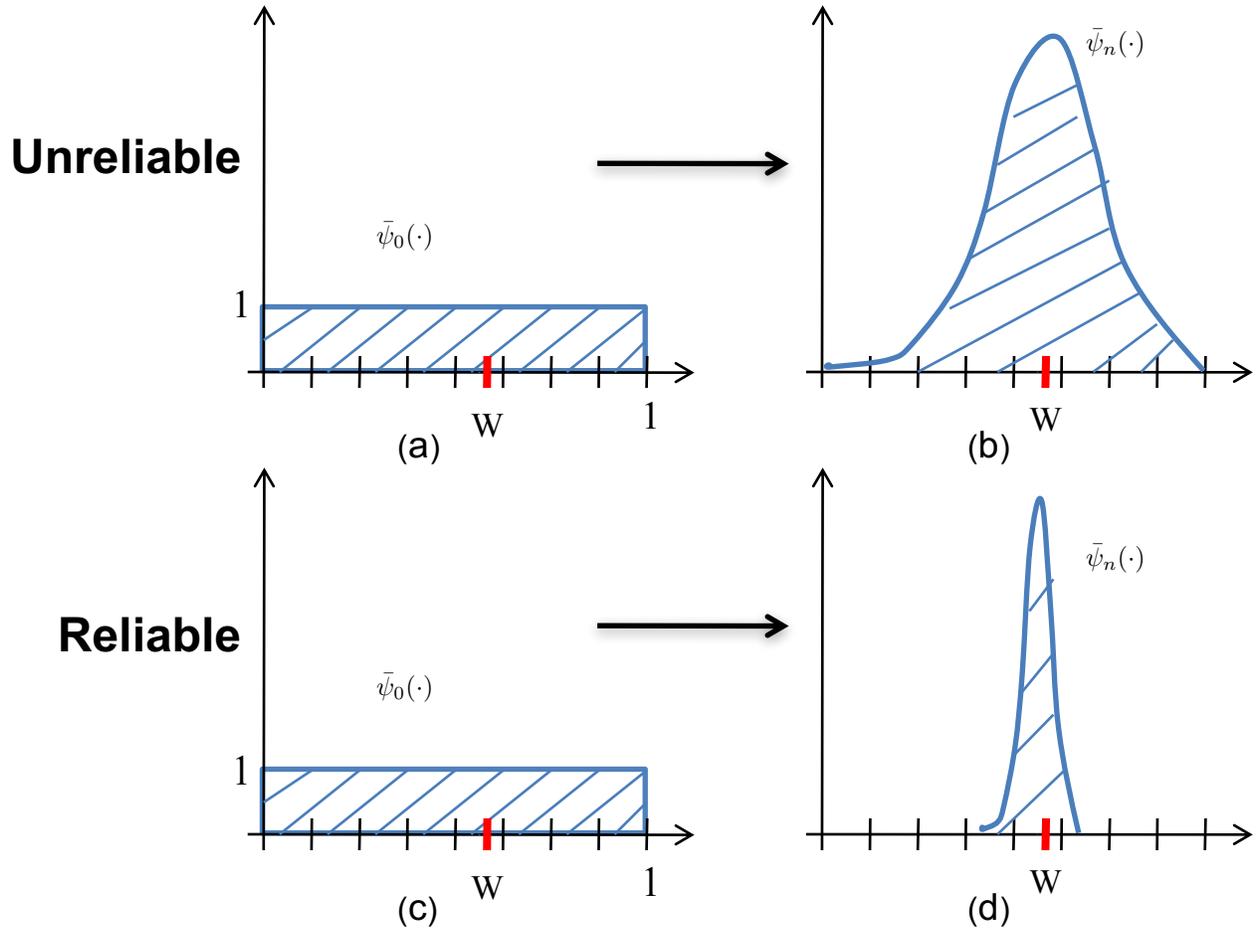

\subsection{Reliable Communication at rate R}\label{sec:sec:problemSetup:achievabilityRateR}
We now provide a definition of achieving a rate $R$ that implies the standard Shannon
theoretic notions in \cite{CoverThomas06}.
\begin{definition}[Rate]\label{defn:achievability:rateR}
    An encoder policy $\encpolicy$ achieves rate $R$ if, given any $\epsilon \in (0,1)$, there
    exists a sequence of Borel sets $(A_n^\epsilon)_{n \geq 1}$ for which:
    \begin{enumerate}
        \item [(i)]  $A_n^\epsilon$ is open and convex for each $n\geq 1$, \;\;$\bP-a.s.$
        \item [(ii)] $\encpolicy$ is reliable.
        \item [(iii)] $\underset{n \to \infty}{\liminf} \bpost_n(A_n^\epsilon) \geq 1-\epsilon \quad \bP-a.s.$
        \item [(iv)] $\underset{n \to \infty}{\liminf} - \frac{1}{n} \log \bnu(A_n^\epsilon) \geq R \quad \bP-a.s.$
    \end{enumerate}
\end{definition}

\begin{remark}
For $\cW=(0,1)$, property (i) is equivalent to $A_n$ being an open interval \cite{rudin1987real}, thus generalizing 
the definition of a {\it decoded interval} in \cite{Shayevitz2011pm} to arbitrary dimensions.  
Whereas achievability of a rate $R$ for PM was defined by Shayevitz \& Feder in terms of the normalized volume and posterior probabilities convergence probability (c.f. \cite[Equation (2)]{Shayevitz2011pm}),  conditions in (iii) and (iv) given here are in terms of almost sure convergence.  Since the Shayevitz \& Feder notion of achievability of a rate $R$ implies that of the standard coding framework \cite[Lemma 3]{Shayevitz2011pm}, and since almost sure convergence implies convergence in probability, Definition~\ref{defn:achievability:rateR} also implies achievability in the standard coding sense.
\end{remark}

We now provide a necessary condition for {\bf any} encoding scheme $\encpolicy$ to be reliable, which leverages recent intrinsic methods in filter stability for hidden Markov models \cite{chigansky2009intrinsic}; this will be used later in the manuscript:
\begin{theorem}\label{thm:reliability:totalvariationPosterior}
  If an encoding scheme $\encpolicy$  is reliable, then
  \[ \lim_{n \to \infty} \E_\nu \brackets{ \norm{\post_n - \bpost_n}} = 0 \]
  for any $\nu \ll \bnu$.
\end{theorem}
\begin{proof}
This follows closely the derivation of equation 1.9 in \cite{chigansky2009intrinsic}.
From \eqref{eqn:posteriorprobabilities}, for any \textit{fixed} measures $\nu$
and $\bnu$ s.t. $\nu \ll \bnu$, $\post_0 = \nu$ and $\bpost_0 = \bnu$,
$\frac{d\P}{d\bP} = \frac{d\nu}{d\bnu}(W).$ From Bayes' rule, for any
non-negative measurable function $g: \cW \to \reals^+$ the following holds,
\begin{align*}
  \E_\nu \brackets{g(W)|\cFYn} &= \frac{\bE \brackets{g(W) \frac{d\P}{d\bP} \Big|\cFYn}}{\bE\brackets{\frac{d\P}{d\bP}\Big|\cFYn}} \\ %
                               &=  \frac{\int_{\cSrc} g(w) \frac{d\nu}{d\bnu}(w) \bpost_n(dw)}{\bE\brackets{\frac{d\nu}{d\bnu}(W)\Big|\cFYn}}  \\ %
  \Rightarrow \frac{d\post_n}{d\bpost_n}(W) &= \frac{\frac{d\nu}{d\bnu}(W)}{\bE\brackets{\frac{d\nu}{d\bnu}(W)\Big|\cFYn}}  %
\end{align*}
Thus we have that the following holds $\P_\nu$ almost surely:
\begin{align}
  \norm{\post_n - \bpost_n} &= \int_{\cW} \abs{\frac{d\post_n}{d\bpost_n}(w) -1} \bpost_n(dw) \nonumber \\
                            &= \frac{ \bE \brackets{ \abs{\RNderivativeA- \denominatorA} \big | \cFYn}}{\denominatorA} \label{eqn:proof:TVlemma:a}
\end{align}
Thus we have that
\begin{align}
  \E_\nu \brackets{\norm{\post_n - \bpost_n}} &=  \bE \brackets{ \RNderivativeA \norm{\post_n - \bpost_n}} \nonumber \\
                                              &=  \bE \brackets{ \denominatorA \norm{\post_n - \bpost_n} } \label{eqn:proof:TVlemma:b}\\
                                              &=  \bE \brackets{ \abs{\RNderivativeA- \denominatorA} } \label{eqn:proof:TVlemma:c}
\end{align}
where \eqref{eqn:proof:TVlemma:b} follows from the tower law of  expectation and
the fact that $\denominatorA$ and $\norm{\post_n - \bpost_n}$ are
$\cFYn$-measurable; and \eqref{eqn:proof:TVlemma:c} follows from
\eqref{eqn:proof:TVlemma:a} and the tower law of expectation. The conditional
expectation $\denominatorA$ in \eqref{eqn:proof:TVlemma:c}  is a nonnegative
uniformly integrable martingale with respect to the increasing filtration
$\cFYn$.  Hence it converges in $L^1(\bP)$ and we have
\begin{align}
  & \lim_{n \to \infty} \E_\nu \brackets{\norm{\post_n - \bpost_n}} \nonumber \\
                                 &=  \bE \brackets{ \abs{\RNderivativeA- \denominatorinf} } \nonumber   \\
                                 &= 0 \label{eqn:proof:TVlemma:d}
\end{align}
where \eqref{eqn:proof:TVlemma:d} follows from the assumption that the PM scheme
is reliable, thus implying that $\denominatorinf = \RNderivativeA$ from Definition~\ref{defn:reliability}.
\end{proof}
 
\section{Posterior Matching Schemes in Arbitrary Dimensions} \label{sec:PMscheme}
In this section, we will introduce a simple feedback-based encoding scheme,
termed the posterior matching (PM) scheme.  The motivation is to satisfy some
necessary conditions on achieving any possible rate. We start by discussing the
properties that any feedback encoding scheme should have for it to maximize
mutual information by looking at the converse to the feedback communication
problem.  We then show how the PM scheme in 1 dimension, developed by
Shayevitz \& Feder in \cite{Shayevitz2011pm}, follows naturally and can be described
in a dynamical system format.  We then define a class of dynamical system encoders 
in arbitrary dimensions and provide a necessary condition on any such encoder to be reliable. 
With this necessary condition, we define posterior matching schemes in arbitrary dimensions,  
 provide unique construction of PM schemes by optimal transport, and showcase a large number of properties that such schemes possess (e.g. stationary, Markov, etc). 

All subsequent discussions assume the following:
\begin{assumption} \label{assump:finiteCapacity}
    The capacity of the communication channel is finite, $I(P_X^*,P_{Y|X}) < \infty$.
\end{assumption}

With this, we have the following Lemma:
\begin{lemma}\label{lemma:finiteMutInfo}
    Under Assumption \ref{assump:finiteCapacity}, $P_{W|Y=y} \ll P_W$ for
    $P_Y$-almost all $y$. Moreover, $i(X;Y)$ is integrable.
\end{lemma}
\begin{proof}
    Now suppose that for some $P_Y$-nontrivial $y$, the absolute continuity
    condition $P_{W|Y=y} \ll P_W$ does not hold. Then clearly by definition,
    $\kldist{P_{W|Y=y}}{P_W} = \infty$ and moreover,
    \begin{align*}
        \infty &= \int_{\cY} \kldist{P_{W|Y=y}}{P_W} P_Y(dy) \\
               &= I(W;Y) \\
               &= I(P_X^*,P_{Y|X})
    \end{align*}
    which is a contradiction. Therefore, according to the definition of the
    information density in \eqref{eqn:defn:informationdensity}, $i(X;Y)$ is
    integrable.
\end{proof}

\subsection{Maximizing Mutual Information}
Given the channel $P_{Y|X}$ and a cost function $\eta: \cX \to \reals_+$, the
capacity-cost function is given by
\beqa
    \capacityCostFn \triangleq \!\!\!\! \max_{\px \in \probSimplex{X}: \E_{\px} \brackets{\eta(X)} \leq L} \!\!\!\! I(\px,\pygx). \label{eqn:defn:capacityCostFn}
\eeqa
It is known that $\capacityCostFn$ is the fundamental limit of communication
over a channel $P_{Y|X}$ of all encoders whose average cost $\eta(X)$ is
upper-bounded by $L$ \cite{mceliece2002theory}.

We note the following standard lemma from information theory \cite{CoverThomas06}:
\begin{lemma} \label{lemma:bound:mutInfo:Capacity}
    Fix an $n \geq 1$.  If a feedback encoder $\parenth{X_i = e_i(W,Y^{i-1}): i =1,\ldots,n}$ satisfies the constraint
    \begin{equation*}
        \bE \brackets{\frac{1}{n}\sum_{i=1}^n \eta(X_i)} \leq L,
    \end{equation*}
    then
    \begin{equation*}
        \frac{1}{n}I(W;Y^n) \leq \capacityCostFn. %
    \end{equation*}  
    Equality holds if and only if
    \bitm
        \item [(a)] $Y_1, \ldots, Y_n$ are statistically independent.
        \item [(b)] $X_i \sim P^*_X$ for each $i$, where $P^*_X$ is the
        capacity-achieving distribution in \eqref{eqn:defn:capacityCostFn}.
    \eitm
\end{lemma}

The original PM scheme for $\cW=(0,1)$ developed by Shayevitz \& Feder is given as follows.
\begin{example}[PM Scheme in One Dimension \cite{Shayevitz2011pm}] \label{example:PMscheme:one-dimension}
Define $F_X(\cdot)$ as the CDF associated with $P_X^*$ above and
\begin{equation}
\begin{aligned}
        X_n &= F_X^{-1}(\tW_n) \quad n \geq 1 \\
        \tW_1 &= W, \quad \tW_{n+1} = F_{W|Y_{1:n}}(W|Y_{1:n}), \quad n\geq 1 \;
         \label{eqn:example:PMscheme:one-dimension:b}
\end{aligned}
\end{equation}
where $F_{W|Y_{1:n}}(w|y_{1:n})$ is the conditional cumulative distribution function (CDF) of the message point $W$ given $n$ channel observations $Y_{1:n}$ under $(\Omega,\cF,\bP)$, e.g. with $\bnu$ as the prior distribution on $W$.
The intuition here is that by placing the message point into its \emph{own} conditional CDF, this guarantees that the new state variable $\tW_{n+1}$ is uniformly distributed, for all $Y_{1:n}$; thus $\tW_{n+1} \sim \bnu$ and is independent of $Y_{1:n}$.  Since $X_{n+1}$ is simply a function of $\tW_{n+1}$, it is still independent of $Y_{1:n}$; moreover, by defining it in terms of the inverse CDF of $P_X^*$, it also has the capacity-achieving distribution $P_X^*$.  As such, this simple scheme guarantees that $I(W;Y_{1:n})=nC$ for all $n \geq 1$.  In addition, $\tW_{n+1}$ is ``handing the decoder what is missing'', because given $Y_{1:n}$, $\tW_{n+1}$ and $W$ form a bijection, by virtue of the monotonicity (and thus invertibility) of CDFs.  It can be shown \cite{Shayevitz2011pm} that the one-dimensional PM scheme \eqref{eqn:example:PMscheme:one-dimension:b} can equivalently be written as
\begin{equation} \label{eqn:PMOneD:encoder-dynamical-system}
\begin{aligned}
        X_n &= F_X^{-1}(\tW_n) \quad n \geq 1 \\
        \tW_1 &= W, \quad \tW_{n+1} = F_{\tW|Y}(\tW_n|Y_n), \;\; n\geq 1 
\end{aligned}
\end{equation}
where $F_{\tW|Y}(\cdot|y)$ is the CDF associated with the posterior distribution $P_{\tW_1|Y_1=y}$ of the message point $\tW_1$ given $Y_1=y$ under $(\Omega,\cF,\bP)$, where $P_W=\bnu$ and $P_{Y|W}(dy|\tw) \equiv P_{Y|X}\parenth{dy|F_X^{-1}(\tw)}$.  
\end{example}

When the one-dimensional posterior matching scheme is written in the format of \eqref{eqn:PMOneD:encoder-dynamical-system}, notice that $(\tW_n)_{n \geq 1}$ is a stochastic dynamical system where $\tW_{n+1} = S_{Y_n}(\tW_n)$ is controlled by the i.i.d. sequence of random variables $(Y_n)_{n \geq 1}$, with $S_y(u) = F_{\tW|Y}(u|y)$.

\subsection{Dynamical System Encoders} \label{sec:PM-scheme-arb-dimension:dynamical-system-encoders}
Inspired by the original one-dimensional PM scheme developed in
\cite{Shayevitz2011pm}, here we also consider dynamical systems $(\tW_n)_{n \geq 1}$ where $\tW_{n+1} = S_{Y_n}(\tW_n)$ and $\tW_1=W$ is the message.  Moving forward, when describing dynamical system encoders, we do not restrict ourselves to $\cW=(0,1)$; rather, we now consider the more general case where only Assumption \ref{assumption:cW} holds.  

\begin{definition}\label{defn:encoder-dynamical-system}
    A dynamical system encoder is a collection of maps $(S_y: \cW \to \cW)_{y \in
    \cY}$ and a memoryless, time-invariant noisy channel with transition kernel
    $P_{Y_n|\tW_n}(\cdot|\cdot) \equiv P_{Y|\tW}(\cdot|\cdot)$. Their dynamics
    govern the random process $(\tW_n,Y_n)_{n \geq 1}$ as follows:
    \begin{equation*}
        \tW_1 = W, \quad \tW_{n+1} = S_{Y_{n}}(\tW_{n}), \quad n\geq 1 
    \end{equation*}
\end{definition}
The intuition here is that $\tW_n$ is signaled into the noisy channel to specify $Y_n$.
At the next step, $\tW_{n+1}$ is governed by $\tW_n$ and the most recent channel output $Y_n$.
As such, it is a time-invariant, memoryless stochastic dynamical system, as shown in Figure~\ref{fig:channelRecast}.
\begin{figure}[t]
    \centering
    \includegraphics[width=0.8\columnwidth]{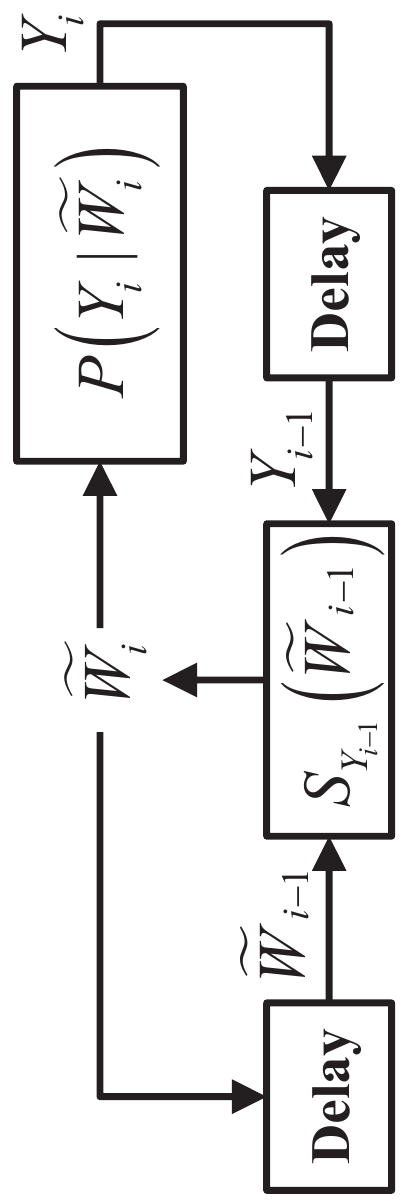}
    \caption{Dynamical system encoder as discussed in
    Definition~\ref{defn:encoder-dynamical-system} that underlies the posterior
    matching scheme.}
    \label{fig:channelRecast}
\end{figure}

The next two lemmas do not make an i.i.d. assumption on $(Y_n)_{n \geq 1}$; rather, they elucidate basic Markov properties any dynamical system encoder possesses, and then provide a necessary condition, for any dynamical system encoder, for recovering the message $W$ from $(Y_n)_{n \geq 1}$. 

\begin{lemma} \label{lemma:MarkovChain:PM:invertibility}
    For any dynamical system encoder,  $(Y_n,\tW_{n+1})_{n \geq 1}$ is a Markov
    process. Moreover, for any $n \geq 1$ the following Markov chain condition holds
    \begin{equation}
        \tW_1 \rightarrow (Y_{1:n},\tW_{n+1}) \rightarrow (Y_{n+1:\infty},\tW_{n+2:\infty}) \label{eqn:lemma:MarkovChain:PM:invertibility:a}
    \end{equation}
\end{lemma}
\begin{proof}
    We note that $\tW_1$ has initial distribution $P_W$. Note that $Y_1$ is
    $\tW_1$ being passed through a noisy memoryless channel, which is a
    stochastic kernel. As such, for appropriately defined functions $g_1$ and
    random variable $N_1$, we have that $Y_1=g_1(\tW_1,N_1)$. Secondly, note
    that $\tW_2=S_{Y_1}(\tW_1) \equiv g_2(\tW_1,N_1)$. This process continues
    onward.

    Because the channel $P_{Y|X}$ is memoryless and time-invariant, this means that $N_1,N_2,\ldots$ are i.i.d.  Thus
    $(Y_n,\tW_{n+1})_{n \geq 1}$ is an iterated function system (IFS) controlled by $(N_n)_{n \geq 1}$.  An IFS is known to be a Markov process \cite{ykifer}.

    As such, we have that the following Markov chain condition holds:
    \[(Y_{1:n-1},\tW_{1:n}) \rightarrow (Y_n,\tW_{n+1}) \rightarrow (Y_{n+1:\infty}, \tW_{n+2:\infty})\]
    We state this equivalently in terms of conditional mutual information and the chain rule:
    \beqa
    0 &=& I \parenth{ Y_{1:n-1},\tW_{1:n} ; Y_{n+1:\infty}, \tW_{n+2:\infty} | Y_n,\tW_{n+1} } \nonumber \\
      &=& I \parenth{ Y_{1:n-1} ; Y_{n+1:\infty}, \tW_{n+2:\infty} | Y_n,\tW_{n+1} } \nonumber \\
      &+& I \parenth{ \tW_1;  Y_{n+1:\infty}, \tW_{n+2:\infty} | Y_{1:n},\tW_{n+1} } \label{eqn:proof:lemma:MarkovChain:PM:invertibility:a}\\
      &+& I \parenth{ \tW_{2:n} ; Y_{n+1:\infty}, \tW_{n+2:\infty} | Y_{1:n},\tW_1,\tW_{n+1} } \nonumber
    \eeqa
    Since conditional mutual information is non-negative, it follows that all three terms must be zero.  \eqref{eqn:proof:lemma:MarkovChain:PM:invertibility:a} being zero is equivalent to the
    Markov chain condition \eqref{eqn:lemma:MarkovChain:PM:invertibility:a}.
\end{proof}

\newcommand{\triangleqEQdiff}[2] {\bE\brackets{\parenth{g(#1) - \bE\brackets{g(#1) \Big|#2}}^2}}

\begin{definition}\label{defn:enc:invertible}
We say a dynamical systems encoder is {invertible} if $\tW_1$ is  $\bP$-a.s. $\sigma(\tW_{n+1}, Y_{1:n})$-measurable $\; \forall n \geq 1$.
\end{definition}
We now show a necessary condition on the structure of $(S_y)_{y \in \cY}$ in order
for an agent only observing $(Y_n)_{n \geq 1}$ to be able to recover the initial condition $\tW_1$.
\begin{lemma} \label{lemma:reliabilityInvertibility}
If a dynamical encoding scheme is reliable, then it is invertible.
\end{lemma}
\begin{proof}
Consider any measurable integrable function $g$.
\newcommand{\EgWgivenBLAH}[1]{\bE \brackets{g(\tW_1) | #1} }
Note then that
\begin{align}
    \EgWgivenBLAH{Y_{1:n},\tW_{n+1}} &= \EgWgivenBLAH{Y_{1:\infty},\tW_{n+1:\infty}} \label{eqn:inv:markov}\\
    &= g(\tW_1) \label{eqn:inv:reliability}
\end{align}
where \eqref{eqn:inv:markov} follows from Lemma~\ref{lemma:MarkovChain:PM:invertibility}; and \eqref{eqn:inv:reliability} follows from the assumption of reliability and Definition~\ref{defn:reliability}, implying that $\tW_1$ is $\bP$-a.s. $\cFYinf$-measurable.

\end{proof}

Note that Lemma~\ref{lemma:reliabilityInvertibility} indicates that in order for 
the message $W$ to be recoverable from $Y_{1:\infty}$, it must be that for any $n \geq 1$, $W$ is recoverable from $\tW_{n+1}$ and the feedback of the past channel outputs $Y_{1:n}$.  This is analogous to a notion of
invertibility defined by Van Handel for a different but related class of
hidden Markov models \cite[Defn 2.6, Remark 2.10]{van2009nonlinear}. This necessary condition on reliability
will motivate our definition of the posterior matching scheme, defined in the next section, to have an invertible structure.

Lemma~\ref{lemma:MarkovChain:PM:invertibility}, we have that $(Y_n,\tW_{n+1})_{n \geq 1}$ is a Markov chain and so we can define a hidden Markov model with state variable $V_n=(Y_n,\tW_{n+1})$ and observation variable $Y_n$.  Note that the posterior distribution on $V_n$ given $Y_{1:n}$ is fully captured by the posterior distribution on $\tW_{n+1}$ given $Y_{1:n}$, which we term  $\pi_n$~($\bpi_n$) under $\P_\nu$~($\bP$), respectively. If a dynamical systems encoder is invertible, then note from Definition~\ref{defn:enc:invertible}  that for any $A \in \cF_\cW$, there exists a $\cFYn$-measurable $B \in \cF_\cW$ for which $\bprob{\tW_1 \in A | \cFYn}=\bprob{\tW_{n+1} \in B | \cFYn}$.  If we replace the roles of $\tW_1$ and $\tW_{n+1}$, we arrive at an analogous statement; thus we have the following remark:
\begin{remark}
    From \eqref{eqn:defn:dTV}, it follows that for any for any invertible dynamical systems encoder: 
        \[
            \|\pi_n-\bpi_n\|=\|\post_n-\bpost_n\|.
        \]
\end{remark}
Note that the fundamental question of filter stability for hidden Markov models involves understanding if the posterior distribution on the state variable at time $n$ becomes insensitive to the prior distribution on the initial state variable \cite[eq 1.5]{chigansky2009intrinsic}:
\[ \lim_{n \to \infty} \E_\nu \brackets{\|\pi_n-\bpi_n\|} \stackrel{?}{=} 0 \quad \forall \nu \ll \bnu.\]
We can combine the necessary condition for {\bf any} encoding scheme $\encpolicy$ to be reliable in Theorem~\ref{thm:reliability:totalvariationPosterior} with the necessary condition for a dynamical systems encoder to be reliable in Lemma~\ref{lemma:reliabilityInvertibility} to conclude that:
\begin{corollary}\label{corollary:VanHandel}
For any invertible dynamical systems encoder, a necessary condition on reliability is that for any $\nu \ll \bnu$,
\[ \lim_{n \to \infty} \E_\nu \brackets{\|\pi_n-\bpi_n\|}= 0.\]
\end{corollary}
This clearly elucidates an intimate connection between the topic of filter stability for hidden Markov models and the topic of reliability for dynamical system encoding schemes for message-point feedback communication problems.

\subsection{The Posterior Matching Scheme}
We now define the posterior matching scheme  for $\cW \subset \reals^d$ as a time-invariant dynamical system encoder that contains all the essential properties of the PM scheme for $\cW=(0,1)$ developed by Shayevitz \& Feder in \cite{Shayevitz2011pm}.
\begin{definition} \label{defn:PM-compatible}
    A posterior matching (PM) scheme is a dynamical system encoder with the following properties:
    \begin{subequations}  \label{eqn:defn:PMlikescheme}
        \begin{align}
            & \tW_1 = W, \quad \tW_{i} = S_{Y_{i-1}}(\tW_{i-1}), \quad i\geq 1  \label{eqn:defn:PMlikescheme:b} \\
            & X_i = \phi(\tW_i) \label{eqn:defn:phi}\\
            & S_y \# P_{\tW_1|Y_1=y} = P_W \text{ invertibly } \quad \forall y \in \cY \label{eqn:defn:PM-compatible:a}, \\
            & \phi \#  P_W =P_X^* \label{eqn:defn:PM-compatible:b}
        \end{align}
    \end{subequations}
\end{definition}
The invertibility of $S_y: \cW \to \cW$ for each $y \in \cY$  is inspired by the
necessary condition of invertibility given in
Lemma~\ref{lemma:reliabilityInvertibility}.  This invertibility property
implies that with knowledge of  $Y_{1:n}$, not only can $\tW_{n+1}$ be
constructed from $\tW_1$, but also that $\tW_1$ can be constructed from $\tW_{n+1}$:
\begin{subequations}\label{eqn:WinvertibleWithWn}
    \begin{align}
        \tW_{n+1} &= S_{Y_{1:n}}(\tW_1) \triangleq  S_{Y_n} \circ \cdots \circ S_{Y_1} (\tW_1) \\
        \tW_1 &=  S^{-1}_{Y_{1:n}}(\tW_{n+1}) \triangleq S_{Y_1}^{-1} \circ \cdots \circ S_{Y_n}^{-1}(\tW_{n+1})
    \end{align}
\end{subequations}

Below is a key lemma which conceptually shows that any PM scheme is ``handing
the decoder what is missing'':
\begin{lemma} \label{lemma:PMscheme:independence}
    For any PM scheme, the following holds:
    \begin{align}
        \bprob{\tW_{n+1} \in \cdot \big| \cFYn} &= \bprob{W \in \cdot}, \quad\bP-a.s. \label{eqn:lemma:PMscheme:independence} \\
        \bprob{\tW_{n} \in \cdot \big| \cFYn}   &= \bprob{\tW_n \in \cdot | Y_n} \label{eqn:lemma:PMscheme:independence:b}
    \end{align}
\end{lemma}
\begin{proof}
    We show this by induction.  First note that for $n=2$, we have that $\tW_2 = S_{Y_1}(\tW_1)$.
    Now we note that given $Y=y$, we have that any function $\xi_y$ satisfying $\tW_2 = \xi_y(\tW_1)$ necessarily pushes
    $P_{\tW_1|Y_1=y}$ to $P_{\tW_2|Y_1=y}$.  In our case, $\xi_y \equiv S_y$, which from \eqref{eqn:defn:PM-compatible:a} pushes $P_{\tW_1|Y_1=y}$ to $P_W$. As such, we have that $P_{\tW_2|Y_1=y} \equiv P_W$.

    Now suppose that for some $n > 2$, we have that
    \beqa
        \bprob{\tW_n \in \cdot \big |\cFYnb} = \bprob{\tW_n \in \cdot} = P_W(\cdot). \label{eqn:proof:lemma:PMscheme:independence:aa}
    \eeqa
     As such, from the time-invariant nature of the memoryless channel, we have that $P_{\tW_n,Y_n}=P_{\tW_1,Y_1}$.  Since $\tW_{n+1}=S_{Y_n}(\tW_n)$, we have that 
    \beqa
        P_{\tW_{n+1}|Y_n=y}=P_{\tW_2|Y_1=y} \equiv P_W \label{eqn:proof:lemma:PMscheme:independence:a}
    \eeqa
    Note that this implies that $I(\tW_{n+1}; Y_n) = 0$ and so from the chain rule, we have that
    \beqa
        I(\tW_n,\tW_{n+1},Y_n; Y_{1:n-1})
        &=& I(\tW_n;Y_{1:n-1}) \nonumber \\
        &+& I(Y_n;Y_{1:n-1}|\tW_n) \nonumber \\
        &+& I(\tW_{n+1}; Y_{1:n-1} | \tW_n,Y_n) \nonumber \\
        &=& 0 \label{eqn:proof:theorem:PM:properties:c}
    \eeqa
    where \eqref{eqn:proof:theorem:PM:properties:c} follows from \eqref{eqn:proof:lemma:PMscheme:independence:aa}, from the memoryless nature of the channel \eqref{eqn:defn:DMC}, and the structure of the dynamical system encoder \eqref{eqn:defn:PMlikescheme:b}.
    Now if we re-write the chain rule in a different order, we have
    \beqa
        0 &=& I(\tW_n,\tW_{n+1},Y_n; Y_{1:n-1}) \nonumber \\
        &=& I(Y_n; Y_{1:n-1})  \nonumber \\
        &+& I(\tW_{n+1}; Y_{1:n-1}|Y_n) \nonumber \\
        &+& I(\tW_n; Y_{1:n-1}| Y_n,\tW_{n+1}) \label{eqn:proof:theorem:PM:properties:d}.
    \eeqa
    From the non-negativity of conditional mutual information and \eqref{eqn:proof:theorem:PM:properties:d}, it follows that $I(\tW_{n+1}; Y_{1:n-1}|Y_n) = 0$.
    Thus, from \eqref{eqn:proof:lemma:PMscheme:independence:a} and \eqref{eqn:proof:theorem:PM:properties:d}, we have that
    \[I(\tW_{n+1}; Y_{1:n}) = I(\tW_{n+1}; Y_n)  + I(\tW_{n+1}; Y_{1:n-1}|Y_n) = 0,\]
    which combined with \eqref{eqn:proof:lemma:PMscheme:independence:a} implies \eqref{eqn:lemma:PMscheme:independence}.

    From the non-negativity of conditional mutual information and \eqref{eqn:proof:theorem:PM:properties:d}, it follows that $I(\tW_{n+1}; Y_{1:n-1}|Y_n)=0$, which implies \eqref{eqn:lemma:PMscheme:independence:b}.
\end{proof}

\begin{definition}
Define $f_n(u)$ and $i_n(u)$ as the posterior density and normalized
information density, respectively:
\begin{subequations}\label{defn:normalizedInfoDensity}
    \begin{eqnarray}
        f_n(u) &\triangleq& \frac{d\bpost_n}{d\bpost_0}(u), \quad u\in \cW \\
        i_n(W) &\triangleq& \frac{1}{n} i(W,Y_{1:n}) = \frac{1}{n} \log f_n(W)
    \end{eqnarray}
\end{subequations}
\end{definition}
\begin{theorem}\label{theorem:PM:properties}
    For any PM scheme,  the following properties hold under $\bP$ for all $n \geq 1$:
    \benum
        \item $(Y_n)_{n \geq 1}$ are $\bP$-i.i.d.
        \item $(\tW_n)_{n \geq 1}$ is a $\bP$-stationary Markov chain.
        \item $(\tW_n,Y_n)_{n \geq 1}$ is a $\bP$-stationary Markov chain.
        \item $\lim_{n \to \infty} i_n(W) = \bE[i(\tW,Y) | \tail_{\tW,Y}] \quad \bP-a.s.$
        \item $I(W;Y_{1:n})=nC$ for any $n \geq 1$.
    \eenum
\end{theorem}
\begin{proof}
    \benum
    \item $I(Y_{n};Y_{1:n-1})=0$ follows from the first term of the sum in \eqref{eqn:proof:theorem:PM:properties:d}.
    Moreover, from \eqref{eqn:lemma:PMscheme:independence} and the time-invariant, memoryless nature \eqref{eqn:defn:DMC} of the channel, it follows that $(Y_n)_{n \geq 1}$ are $\bP$-i.i.d.
    \item From 1), $(Y_n)_{n \geq 1}$ are $\bP$-i.i.d., and so the dynamical system encoder
     \eqref{eqn:defn:PMlikescheme:b} corresponds to an iterated
    function system $(\tW_n)_{n \geq 1}$ controlled by $(Y_n)_{n \geq 1}$.  This
    implies that $(\tW_n)_{n \geq 1}$ is a Markov chain \cite{ykifer}.  It is stationary from \eqref{eqn:lemma:PMscheme:independence}.
    \item From \eqref{eqn:defn:PMlikescheme:b} in the PM scheme definition, $P_{\tW_{n+1}|\tW_{1:n},Y_{1:n}}$ is a point mass at $S_{Y_n}(\tW_n)$. Moreover, from \eqref{eqn:defn:phi} in the PM scheme definition and \eqref{eqn:defn:DMC}, we have that $P_{Y_{n+1}|Y_{1:n},\tW_{1:n+1}} \parenth{dy_{n+1}|y_{1:n},\tw_{1:n+1}}$ $=P_{Y|X}\parenth{dy_{n+1}|\phi(\tw_{n+1})}$.  Thus $(\tW_n,Y_n)_{n \geq 1}$ is Markov.  Stationarity again follows from  \eqref{eqn:lemma:PMscheme:independence}.

    \item From 3), $(\tW_n,Y_n)_{n \geq 1}$ is a stationary
    Markov chain with invariant measure $P_{W,Y}$. Thus we have that
    \begin{align}
         i_m(W) &= \frac{1}{m} \log \frac{d\bpost_m}{d\bpost_0}(W) \label{eqn:lemma:IIFS:BirkhoffT:aa}  \\
          &= \frac{1}{m} \sum_{n=1}^m \log \frac{d\bpost_n}{d\bpost_{n-1}}(W) \nonumber \\
          &= \frac{1}{m} \sum_{n=1}^m \log \frac{d \bprob{\tW_1 \in \cdot | \cFYn}}{ d\bprob{\tW_1 \in \cdot | \cFYnb}}(W) \nonumber \\
          &= \frac{1}{m} \sum_{n=1}^m \log \frac{d \bprob{\tW_n \in \cdot | \cFYn}}{ d\bprob{\tW_n \in \cdot | \cFYnb}}(\tW_n) \label{eqn:lemma:IIFS:BirkhoffThm:a} \\
          &= \frac{1}{m} \sum_{n=1}^m \log \frac{d \bprob{\tW_n \in \cdot | Y_n}}{ d\bprob{\tW_n \in \cdot}}(\tW_n) \label{eqn:lemma:IIFS:BirkhoffThm:b}  \\
          &=  \frac{1}{m} \sum_{n=1}^m i(\tW_n,Y_n) \label{eqn:lemma:IIFS:BirkhoffThm:c}  \\
          &\to \bE[ i(\tW_1,Y_1) | \tail_{\tW,Y}] \label{eqn:lemma:IIFS:BirkhoffThm:d}
    \end{align}
    where \eqref{eqn:lemma:IIFS:BirkhoffT:aa} follows from
    \eqref{defn:normalizedInfoDensity}; \eqref{eqn:lemma:IIFS:BirkhoffThm:a}
    follows from \eqref{eqn:WinvertibleWithWn}: given $Y_{1:n-1}$, $\tW_1$ and
    $\tW_n$ are in one-to-one correspondence;
    \eqref{eqn:lemma:IIFS:BirkhoffThm:b} follows from both conditions in
    Lemma~\ref{lemma:PMscheme:independence};
    \eqref{eqn:lemma:IIFS:BirkhoffThm:c} follows from the definition
    \eqref{eqn:defn:informationdensity} of the information density, and
    \eqref{eqn:lemma:IIFS:BirkhoffThm:d} follows from Birkhoff's pointwise
    ergodic theorem applied to the stationary Markov process $(\tW_n,Y_n)_{n
    \geq 1}$ where the integrability of the random variable $i(W_1,Y_1)$ with
    respect to $P_{W,Y}$ follows from Lemma~\ref{lemma:finiteMutInfo}.

    \item This follows from Lemma~\ref{lemma:bound:mutInfo:Capacity} since (a)
    $(Y_n)_{n \geq 1}$ are i.i.d. and (b) $\tW_n$ has distribution $P_W$, thus
    implying from \eqref{eqn:defn:PM-compatible:b} that each $X_n \sim P_X^*$.
    \eenum
\end{proof}

This leads to the following corollary:
\begin{corollary} \label{cor:PM-compatible-phi-invertible}
    If $\phi$ is invertible, then the posterior matching scheme can be
    equivalently described as
    \begin{align*}
        & X_1 = \phi(W), \quad X_{i+1} = \phi \circ S_{Y_i} \circ \phi^{-1}(X_i) \\
        & (\phi \circ S_y \circ \phi^{-1}) \# P_{X_1|Y_1=y}=P_{X}^*
    \end{align*}
\end{corollary}

\subsection{Unique Construction of PM schemes in Arbitrary Dimensions} \label{sec:PMOTTheory}
We now note that there are multiple possible PM schemes for a given $\pygx$ and $\cW$, even in one dimension. 
Since for any uniform $(0,1)$ random variable $V$, the random variable $1-V$ is also uniform $(0,1)$, consider \eqref{eqn:example:PMscheme:one-dimension:b} in Example~\ref{example:PMscheme:one-dimension} and replace it with
\[        \tW_1 = W, \quad \tW_{n+1} = 1-F_{W|Y_{1:n}}(W|Y_{1:n}), \quad n\geq 1.  \]
This is clearly also a PM scheme.  In larger dimensions, there are \emph{many} possible PM schemes and so it would be desirable to always select a \emph{unique} PM scheme. We here demonstrate a way to do so using optimal transport theory (OTT), which involves finding an optimal mapping that transforms samples from one distribution to another, under an appropriate measure of cost.  For $\cU,\cV \subset \reals^d$, consider a cost function $c(u,v)$ on $\cU \times \cV$.
\begin{definition}\label{defn:OTTproblem}
Given a cost function $c: \cU \times \cV \to \reals$ and a pair of distributions $P_U \in \probSimplex{\cU},P_V \in \probSimplex{\cV}$, we 
consider the following optimization problem:
\beqa 
\OTTcost(P_U,P_V,c): \inf_{S: \;\;S \# P_U=P_V} \int_{\cU } c(u,S(u)) P_U(du) \label{defn:eqn:OptimalTransportProblem}
\eeqa
\end{definition}
Monge was the first to formulate this problem \cite{monge1781}, while Kantorovich reformulated \eqref{defn:eqn:OptimalTransportProblem}
to a more general problem of optimization over a space of joint distributions that preserves marginals $P_U$ and $P_V$ \cite{kantorovich1942mass}. This problem has
also been studied in depth in \cite{villani2003topics,villani2009optimal, rachev1998mass}.
A standard and well-studied version is under the quadratic cost.
\begin{definition}
For a $M \in \reals^{d \times d}$ for which $M \succ 0$, we define the following cost functions of the form $c_M(u,v) \triangleq \|u-v\|_M^2 \triangleq (u-v)^TM(u-v)$.
\end{definition}

\subsubsection{Existence of $\phi$}
We now demonstrate that for appropriate cost functions $c(u,v)$, such a map $\phi$ satisfying
property \eqref{eqn:defn:PM-compatible:b} can be recovered from solving $\OTTcost(P_W,P_X^*,c)$.   
\begin{theorem}\label{thm:monge:phi}
    Let $c(u,v)=\norm{u-v}$ where $\norm{\cdot}$ is a strictly convex norm on $\reals^d$. Then the problem $\OTTcost(P_W,P_X^*,c)$ 
    has at least one optimal solution.
\end{theorem}
\begin{proof}
This follows directly from
\cite[Theorem 1.1]{champion2010monge} where we simply exploit the fact that clearly $P_W$, the uniform distribution on $\cW$, has a density with respect to the Lebesgue measure (e.g. $P_W \ll \mu$). 
\end{proof}

\begin{remark}
This theorem's significance is that a transformation $\phi$ satisfying \eqref{eqn:defn:PM-compatible:b} can be found with optimal transport theory, even if $P_X^*$ does not have a density with respect to the Lebesgue measure.  
\end{remark}

For example, suppose $\cW \subset \reals^d$ and $P_{Y|X}$ is a discrete memoryless channel.  Here, we simply let $\cX \subset \reals^d$
and specify  $P_X^*$ to place atoms at the countable set of points in $\cX$ with associated atom probabilities from the optimal input distribution.

\subsubsection{Uniqueness of $\phi$}
We can now state another standard theorem from optimal transport theory, pertaining to when the cost function $c_M$ with $M=I$, e.g. $c_I(u,v)=\|u-v\|^2$:
\begin{theorem}[\cite{villani2009optimal}, Theorem 9.4]\label{thm:monge:unique}
    If $P_U \ll \mu$ and $P_U,P_V$ both have finite second moments,
    then the problem $\OTTcost(P_U,P_V,c_I)$ has a unique optimal solution.
\end{theorem}
Note that since $P_W \ll \mu$ and from Assumption~\ref{assumption:cW}, $P_W$ has finite second moment. We have the following corollary as a consequence of Theorem~\ref{thm:monge:unique}:
\begin{corollary} \label{corollary:monge:unique:phi}
    If $P_X^*$ has finite second moment, then the problem $\OTTcost(P_W,P_X^*,c_I)$ has a unique optimal solution $\phi$ which satisfies property \eqref{eqn:defn:PM-compatible:b}.
\end{corollary}

\subsubsection{Uniqueness of $S_y$}
In order to satisfy the necessary condition from Lemma~\ref{lemma:reliabilityInvertibility}, the PM scheme 
involves an  {\it invertible} map $S_y: \cW \to \cW$ satisfying
\eqref{eqn:defn:PM-compatible:a}.  We now demonstrate that such a map can be
uniquely and explicitly constructed with optimal transport theory.  We do this
by specifying a weighted quadratic cost for the Monge-Kantorovich problem pertaining
to Brenier's problem.  
\begin{theorem}[Generalized Brenier's Theorem]%
    Suppose $\cW \subset \reals^d$, $P, Q \in \probSimplex{\cW}$, and $P,Q \ll \mu$ which induce
    densities $p(u)$ and $q(v)$ respectively with respect to
    the Lebesgue measure $\mu$.  For any $M \succ 0$, consider the following problem $\OTTcost(P,Q, c_M)$.
    Then there exists a unique $S^*$, which is a diffeomorphism, that attains the optimal cost.
\end{theorem}

\begin{proof}
Note that when $M=I_{d \times d}$, the identity matrix, this is the standard Brenier
theorem, whose proof can be found in  \cite[Theorem 2.1.5]{villani2003topics}.
Now we generalize for arbitrary $M \succ 0$, for which we can express $M = U
\Lambda U^T$ where $UU^T=I$. Let $B=(U\sqrt{\Lambda})^T$, $\tu=B u, \tv=B v$,
and $u = B^{-1} \tu, v = B^{-1} \tv$, then
\begin{eqnarray}
\|u-v\|_M^2 &=&(u-v)^T M(u-v) \nonumber \\
       &=&(Bu - Bv)^T(Bu - Bv) \nonumber \\
       &=&\|Bu - Bv\|_I^2 \nonumber
\end{eqnarray}
\newcommand{\tilP}{\tilde{P}}
By defining $S_B(u)= Bu$, then note that:
\begin{itemize}
\item if $U \sim P_U$, then $BU \sim S_B \# P_U$
\item if $S(U) \sim P_V$, then $B S(U) \sim S_B \# P_V$.
\end{itemize}
By defining $\tilP_U \triangleq S_B \# P_U$ and $\tilP_V \triangleq S_B \# P_V$, 
\beqa 
&& \inf_{S \# P_U=P_V} \int_{\cU } \|u-S(u)\|_M^2 P_U(du) \\
&=& \inf_{S \# P_U=P_V} \int_{\cU } \|Bu-BS(u)\|^2 P_U(du) \\
&=& \inf_{\tS \# \tilP_U=\tilP_V} \int_{\cU } \|\tu-\tS(\tu)\|^2 \tilP_U(d \tu)
\eeqa
and so we have that the unique optimal solution to $\OTTcost(P,Q, c_M)$ is  the unique optimal solution to
$\OTTcost(\tilde{P},\tilde{Q}, c_I)$, and moreover is a diffeomorphism. 
\end{proof}

Note that Assumption \ref{assump:finiteCapacity} implies that $P_{W|Y=y} \ll \bnu$.  Since $P_W \ll \bnu$, we have the following corollary:
\begin{corollary}\label{corollary:unique:S}
Under  Assumption \ref{assump:finiteCapacity},  $\OTTcost(P_{W|Y=y},P_W, c_M)$ has a unique optimal solution $S_y$ that is invertible for  any $M \succ 0$, and $P_Y$-almost all $y$.
\end{corollary}

As we will see in Sections~\ref{subsec:symbren},\ref{subsec:knhd}, for $d > 1$, different positive
definite matrices $M$ give rise to different maps $S_y$ solving $\OTTcost(P_{W|Y=y},P_W, c_M)$ that satisfy
    \eqref{eqn:defn:PM-compatible:a}.

\section{Reliability of PM}\label{sec:reliability}

The PM scheme in Definition~\ref{defn:PM-compatible} maximizes the mutual
information $I(W;Y^n)=nC$ \cite{Shayevitz2011pm,ma2011generalizing}. However, it
need not be reliable in general, as shown in \cite[Example 11]{Shayevitz2011pm}.
`Fixed-point-free' necessary conditions are given in
\cite[Lemma 21]{Shayevitz2011pm}.  Our objective here is to develop general {\bf
necessary and sufficient} conditions for PM reliability. 

We next provide an applied probability result that will be used throughout.  
\begin{lemma} \label{lemma:equivalentConditions:MarkovStationaryErgodic}
  Consider a measurable space $(\cV, \cF_\cV)$ and a time-homogeneous Markov process $(V_n)_{n
  \geq 1}$,  where $\bnu$ defined on $(\cV, \cF_\cV)$ is the invariant
  distribution and  $(V_n)_{n \geq 1}$ is stationary on $(\Omega,\cF,\bP)$.  
  Then (iii) $\Rightarrow$ (ii) $\Leftrightarrow$ (i) $\Rightarrow$ (iv) where:
  \begin{itemize}
    \item [i)] $(V_n)_{n \geq 1}$ is $\bP$-ergodic.
    \item [ii)] For any $D \in \cF$:
    \begin{equation}
      \lim_{n \to \infty} \sup_{E_n \in \sigma(V_{n:\infty})} \abs{\bP(E_n \cap D)- \bP(E_n)\bP(D)} = 0. \label{eqn:lemma:equivalentConditions:MarkovStationaryErgodic:a}
    \end{equation}
    \item [iii)] $\norm{\P_\nu(V_{n} \in \cdot )-\bnu} \to 0$ for any $\nu \ll \bnu$,  where $\P_\nu$ is the distribution on the Markov process  $(V_n)_{n \geq 1}$ for which
  $V_1 \sim \nu$, e.g. $\frac{d \P_\nu}{d \bP} = \frac{d \nu}{ d\bnu}(V_1)$. 
  \item [iv)]
  For any separable $\cH \subset \cF$:
\beqa
\sup_{E \in \tail_\cV} \abs{\bprob{E|\cH} - \bprob{E}}=0 \quad \bP-a.s. \label{eqn:lemma:equivalentConditions:MarkovStationaryErgodic:conditional}
\eeqa
  \end{itemize}
\end{lemma}
\begin{proof}
  Consider any $I \in \cF_{\cV}$, for which $\bnu(I) > 0$ and define $\nu$ to have atoms $(I,I^c)$, e.g.
  \begin{equation}
    \frac{d\nu}{d\bnu}(V_1) = \frac{\indicatorvbl{V_1 \in I}}{\bnu(I)}. \label{eqn:proof:lemma:equivalentConditions:MarkovStationaryErgodic:a}
  \end{equation}
  Since for any $\P \ll \bP$, $\E[g(W)]=\bE\brackets{g(W) \frac{d\P}{d\bP}}$
  we have that
  \begin{align}
    \E_\nu \brackets{\indicatorvbl{\tW_n \in J}} &= \bE \brackets{\indicatorvbl{\tW_n \in J} \frac{\indicatorvbl{W \in I}}{\bnu(I)}} \nonumber \\
    \Leftrightarrow \bP \parenth{\tW_n \in J, \tW_1 \in I} &= \bnu(I) \P_\nu\parenth{\tW_n \in J}  \label{eqn:proof:lemma:equivalentConditions:MarkovStationaryErgodic:b} 
  \end{align}

  We now show that (iii) $\Rightarrow$ (i).  If (iii) holds, then from \eqref{eqn:proof:lemma:equivalentConditions:MarkovStationaryErgodic:b} we have that for any $I \in \cF_\cW$ such that $\bnu(I)>0$:
  \begin{align}
  \lim_{n \to \infty} \bP \parenth{\tW_n \in J, \tW_1 \in I} &= \lim_{n \to \infty}  \bnu(I) \P_\nu\parenth{\tW_n \in J} \nonumber \\
                                                             &= \bnu(I)\bnu(J). \label{eqn:proof:lemma:equivalentConditions:MarkovStationaryErgodic:fff}
  \end{align}
where \eqref{eqn:proof:lemma:equivalentConditions:MarkovStationaryErgodic:fff} follows from \eqref{eqn:proof:lemma:equivalentConditions:MarkovStationaryErgodic:b} and condition (iii).  If $\bnu(I) = 0$, then clearly \eqref{eqn:proof:lemma:equivalentConditions:MarkovStationaryErgodic:fff} holds.  Thus under $\bP$, the strictly stationary random process $(V_n)_{n\geq 1}$ is mixing (in the ergodic theory sense \cite[Sec 2.5]{bradley2005basic}).  
From Lemma~\ref{corollary:mixingImpliesErgodicity}, $(V_n)_{n\geq 1}$ is $\bP$-ergodic and so (i) follows.

(i) $\Leftrightarrow$ (ii) follows from \cite[Thm 2]{blackwell1964tail}. 

We now show that (ii) $\Rightarrow$ (iv).  We borrow ideas from
\cite[Sec 2.2]{van2007lecture}.
\newcommand{\cHm}{\cH^{(m)}}
Let $A \in \tail_\cV$. 
Since $\cH$ is separable, note that
\[\cH = \sigma(\cHm: m \geq 1)\]
where each $\cHm$ is finite.   Assuming $\cHm$ is generated by the partition $\braces{B_k: k=1,\ldots,K_m}$, we have that \cite[Defn 2.1.2]{van2007lecture}:
\beqas
\bprob{A|\cHm} &=&  \sum_{k=1}^{K_m} \bprob{A|B_k} \indicatorvbl{B_k} \\
               &=&  \sum_{k=1}^{K_m} \frac{\bprob{A \cap B_k}}{\bprob{B_k}} \indicatorvbl{B_k}.
\eeqas
Thus we have that
\beqa
&& \abs{\bprob{A|\cHm}-\bprob{A}} \nonumber \\
  &=& \abs{\sum_{k=1}^{K_m} \frac{\bprob{A \cap B_k} - \bprob{A}\bprob{B_k}}{\prob{B_k}} \indicatorvbl{B_k}} \nonumber \\
  &\leq& \sum_{k=1}^{K_m} \frac{\abs{\bprob{A \cap B_k} - \bprob{A}\bprob{B_k}}}{\prob{B_k}} \\
  &\leq& \!\!\sum_{k=1}^{K_m}\! \sup_{E_n \in \sigma(V_{n:\infty})} \!\!\!\frac{\abs{\bprob{E_n \cap B_k} - \bprob{E_n}\bprob{B_k}}}{\prob{B_k}} 
  \label{eqn:proof:lemma:equivalentConditions:MarkovStationaryErgodic:conditional:a} \\
&=& 0 \label{eqn:proof:lemma:equivalentConditions:MarkovStationaryErgodic:conditional:b} 
\eeqa
where \eqref{eqn:proof:lemma:equivalentConditions:MarkovStationaryErgodic:conditional:a}  follows since $A \in \tail_\cV$ and $\tail_\cV \subset \sigma(V_{n:\infty})$ for each $n \geq 1$; and \eqref{eqn:proof:lemma:equivalentConditions:MarkovStationaryErgodic:conditional:b}  follows from \eqref{eqn:lemma:equivalentConditions:MarkovStationaryErgodic:a} and that $\abs{\bprob{A|\cHm}-\bprob{A}}$ has no dependence upon $n$. Since $\bprob{A|\cH}=\lim_{m \to \infty} \bprob{A|\cHm}$, it follows that
\beqa
\abs{\bprob{A|\cH}-\bprob{A}} &=& \lim_{m \to \infty} \abs{\bprob{A|\cHm}-\bprob{A}} \nonumber \\
                              &=& 0. \quad \bP-a.s. \label{eqn:lemma:equivalentConditions:MarkovStationaryErgodic:c}
\eeqa
Since \eqref{eqn:lemma:equivalentConditions:MarkovStationaryErgodic:c}  holds for any $A \in \tail_V$, it follows that
\[\sup_{E \in \tail_\cV} \abs{\bprob{E|\cH} - \bprob{E}}=0.\quad \bP-a.s. \]
\end{proof}

\begin{theorem} \label{thm:necessarySufficientConditionsPM:reliability}
The PM scheme is reliable if and only if $(\tW_n)_{n \geq 1}$ is $\bP$-ergodic.
\end{theorem}
\begin{proof}
If the PM scheme is reliable, then for any $\nu \ll \bnu$,
\begin{align}
&\norm{ \P_\nu(\tW_{n+1} \in \cdot ) - \bnu} \nonumber \\
&= \norm{ \E_\nu \brackets{\P_\nu(\tW_{n+1} \in \cdot |\cFYn) - \bnu}}
                                        \label{eqn:thm:necessarySufficientConditionsPM:reliability:aa}  \\
                                        &= \norm{ \E_\nu \brackets{\P_\nu(\tW_{n+1} \in \cdot |\cFYn) - \bP(\tW_{n+1} \in \cdot | \cFYn)}}
                                        \label{eqn:thm:necessarySufficientConditionsPM:reliability:bb} \\
                                        &= \norm{ \E_\nu \brackets{\P_\nu(\tW_{1} \in \cdot |\cFYn) - \bP(\tW_{1} \in \cdot | \cFYn)} }
                                        \label{eqn:thm:necessarySufficientConditionsPM:reliability:cc} \\
                                        &= \norm{ \E_\nu \brackets{\post_n - \bpost_n} } \nonumber \\
                                       &\leq \E_\nu \brackets{\norm{\post_n - \bpost_n}} \label{eqn:thm:necessarySufficientConditionsPM:reliability:dd} \\
                                        &\to 0 \label{eqn:thm:necessarySufficientConditionsPM:reliability:ff}
\end{align}
where \eqref{eqn:thm:necessarySufficientConditionsPM:reliability:aa} follows from the tower law of expectation;
\eqref{eqn:thm:necessarySufficientConditionsPM:reliability:bb} follows from Lemma~\ref{lemma:PMscheme:independence};
\eqref{eqn:thm:necessarySufficientConditionsPM:reliability:cc} follows from the invertibility of the $S_{Y^n}$ map under the PM from Definition~\ref{defn:PM-compatible};
\eqref{eqn:thm:necessarySufficientConditionsPM:reliability:dd} follows from Jensen's inequality;
and \eqref{eqn:thm:necessarySufficientConditionsPM:reliability:ff} follows from the assumption that the PM scheme is reliable and Theorem~\ref{thm:reliability:totalvariationPosterior}.   
Thus condition (iii) in Lemma~\ref{lemma:equivalentConditions:MarkovStationaryErgodic} holds.
From (iii) $\Rightarrow$ (i) in Lemma~\ref{lemma:equivalentConditions:MarkovStationaryErgodic}, it follows that $(\tW_n)_{n \geq 1}$  is $\bP$-ergodic.

\newcommand{\ysequence}{y_{1:\infty}}
\newcommand{\Vny}{B_n}
\newcommand{\CLone}{\mathcal{CL}_1}
\newcommand{\Any}{L_n}
\newcommand{\Bny}{M_n}

\newcommand{\Yn}{Y_{1:n}}
\newcommand{\yn}{y_{1:n}}

\newcommand{\bAn}{\overline{E}_{n+1}}
\newcommand{\uAn}{\underline{E}_{n+1}}
\newcommand{\Ayn}{E_{n+1}^{\yn}}
\newcommand{\AYn}{E_{n+1}^{\Yn}}
\newcommand{\bMn}{\overline{M}_n}
\newcommand{\uMn}{\underline{M}_n}

 \newcommand{\cGnB}{\cG^{(n)}_B}

\newcommand{\tailW}{\tail_{\tW}}
Now suppose that $(\tW_n)_{n \geq 1}$ is $\bP$-ergodic.  
\newcommand{\cFn}{\cF^{(n)}}

Let $y=(y_1,y_2,\ldots) \in \cY^\infty$ and define
\[A_y \triangleq \bigcap_{n \geq 1} \bigcup_{m \geq n} \braces{\tW_{m+1} \in S_{y_{1:m}}(I)}.\]
Note that clearly $A_y \in \tailW$ and so from ergodicity,
\beqa
 \bP(A_y) \in \braces{0,1} \quad \forall \; y \in \cY^\infty. \label{eqn:ProbAy:zeroOne}
\eeqa
Moreover, note that
\beqas
A_y &=& \bigcap_{n \geq 1} \bigcup_{m \geq n} \braces{\tW_{m+1} \in S_{y_{1:m}}(I)} \\
    &=& \bigcap_{n \geq 1} \bigcup_{m \geq n} \braces{W \in T_{Y_{1:m}} \circ S_{y_{1:m}}(I)} 
\eeqas
Thus if $Y(\omega)=y$, then we have that
\begin{align}
\indicatorvbl{A_y}(\omega) &= \indicatorvbl{\bigcap_{n \geq 1} \bigcup_{m \geq n} \braces{W \in T_{Y_{1:m}(\omega)} \circ S_{y_{1:m}}(I)}}(\omega) \nonumber \\
&= \indicatorvbl{\bigcap_{n \geq 1} \bigcup_{m \geq n} \braces{W \in T_{y_{1:m}} \circ S_{y_{1:m}}(I)}}(\omega) \label{eqn:proof:reliability:Sep2018:a} \\
&=\indicatorvbl{\bigcap_{n \geq 1} \bigcup_{m \geq n} \braces{W \in I}}(\omega) \nonumber\\
&= \indicatorvbl{W \in I}(\omega) \label{eqn:proof:reliability:Sep2018:b}
\end{align}
where in \eqref{eqn:proof:reliability:Sep2018:a}, we exploit the fact that $Y_{1:m}(\omega)=y_{1:m}$.
As such, if $Y(\omega)=y$ then
\beqas
\bprob{A_y|\cFYinf}(\omega) &=& \bE\brackets{\indicatorvbl{A_y}(\omega) | \cFYinf}(\omega) \\
                            &=& \bE\brackets{\indicatorvbl{W \in I}(\omega) | \cFYinf}(\omega) \\
                            &=& \bprob{W \in I |\cFYinf}(\omega).
\eeqas

Note that $\bprob{W \in I |\cFYinf}$ is any $\cFYinf$-measurable function such that 
for any $B \in \cB(\cY^\infty)$,
\[\bprob{W \in I,Y \in B} = \int_{\omega \in Y^{-1}(B)} \bprob{W \in I |\cFYinf} \bP(d\omega). \]
Define $A_Y$ as $A_y$ where $Y(\omega)=y$. As such,
\beqas
\bP(A_Y|\cFYinf)&=&\bP(A_y|\cFYinf) \quad \text{ when } Y(\omega)=y \\
\bP(A_Y)&=&\bP(A_y) \quad \text{ when } Y(\omega)=y.
\eeqas
Thus $\bP(A_Y)$ and $\bP(A_Y|\cFYinf)$ are both $\cFYinf$-measurable random variables.  So for any $B \in \cB(\cY^\infty)$,
\beqa
&& \abs{\bprob{W \in I,Y \in B}-\int_{Y^{-1}(B)}\bprob{A_Y}\bP(d\omega)}  \nonumber \\
&=& \abs{\int_{Y^{-1}(B)} \bprob{W \in I|\cFYinf}\bP(d\omega)- \bprob{A_Y}\bP(d\omega) } \nonumber\\
&\leq& \int_{Y^{-1}(B)} \abs{\bprob{W \in I|\cFYinf}- \bprob{A_Y}}\bP(d\omega) \nonumber\\
&=& \int_{Y^{-1}(B)} \abs{\bprob{A_Y|\cFYinf}- \bprob{A_Y}}\bP(d\omega) \label{eqn:proof:reliability:Nov2018:a}\\
&\leq& \int_{Y^{-1}(B)} \sup_{A \in \tailW} \abs{\bprob{A|\cFYinf}- \bprob{A}} \bP(d\omega) \nonumber \\
&=& 0.  \label{eqn:proof:reliability:Nov2018:b}
\eeqa
where \eqref{eqn:proof:reliability:Nov2018:a} follows from \eqref{eqn:proof:reliability:Sep2018:b} and
 \eqref{eqn:proof:reliability:Nov2018:b} follows 
from (i) $\Rightarrow$ (iv) in  Lemma~\ref{lemma:equivalentConditions:MarkovStationaryErgodic}. 

It thus follows from \eqref{eqn:proof:reliability:Nov2018:b} and \eqref{eqn:ProbAy:zeroOne} that
\beqa
\bprob{W \in I | \cFYinf}=\bprob{A_Y} \in \braces{0,1} \quad \bP-a.s.
\eeqa
and thus $\bprob{W \in I | \cFYinf}=\indicatorvbl{W \in I}\;\;\bP$-a.s.

Any measurable function $g:\cW \to \reals$ is the pointwise limit of a sequence simple functions:
$g = \lim_{m \to \infty} g_m$.  Any such simple $g_m:\cW \to \reals$ can be expressed in terms of a partition $(B_1,\ldots,I_{N_m})$ of $\cW$ as  $g_k(w)=\sum_{i=1}^{N_m} a_i \indicatorvbl{w \in I_i}$.  Clearly, we have that $\bE[g_k(W) | \cFYn] \to g_k(W)$ for any $k$.  Thus  $\bE[g(W)|\cFYn] \to g(W)$ for any measurable $g$ and so from Definition~\ref{defn:reliability} the PM scheme is reliable.

\end{proof}
 
\section{Achieving any Rate $R < C$} \label{sec:achievingCapacity}
\newcommand{\hWMAP}{\hat{W}_{\text{MAP}}}
\newcommand{\fnmax}{f_n^{\max}}
\newcommand{\AnW}{A_n^W}
\newcommand{\BnW}{B_n^W}
\newcommand{\InW}{I_n^W}
\renewcommand{\tB}{B}
\newcommand{\liminfn}{\underset{n \to \infty}{\liminf \;}}
\newcommand{\RNposteriorsrst}[3]{\frac{d \bpost_{#1|#3}}{d \bpost_{#2|#3}}}
\newcommand{\RNposteriors}[2] {\frac{d \bpost_{#1}}{d \bpost_{#2}}}
\newcommand{\RNconsecposteriors}[1] {\frac{d \bpost_{#1}}{d \bpost_{#1-1}}}
\newcommand{\RNconsecposteriorsrst}[2] {\RNposteriorsrst{#1}{#1-1}{#2}}
\newcommand{\cFYi}{\cF^Y_{1:i}}
\newcommand{\kldistposteriorsconsec}[1]{\kldist{\bpost_{#1}}{\bpost_{#1-1}}}
\newcommand{\kldistposteriorsconsecrst}[2]{\kldist{\bpost_{#1|#2}}{\bpost_{#1-1|#2}}}
\newcommand{\limn}{\underset{n \to \infty}{\lim}}
\newcommand{\bEba}[1]{\bE\brackets{\abs{#1}}}
\newcommand{\ti}{\tilde{i}}

In this section, we will establish several definitions and lemmas aimed at constructing the main theorem of this section: Theorem~\ref{theorem:equivalenceThreeConditionsPM:reliable:ergodic:capacity}, establishing the equivalence between ergodicity, reliability, and achievability.

\begin{lemma} \label{lemma:MarginalErgodicityImpliesJointErgodicity}
 The random process $(\tW_n)_{n \geq 1}$ is $\bP$-ergodic if and only if $(\tW_n,Y_n)_{n \geq 1}$ is $\bP$-ergodic.
\end{lemma}
\begin{proof}
It trivially  follows that $(\tW_n,Y_n)_{n \geq 1}$ being $\bP$-ergodic implies that $(\tW_n)_{n \geq 1}$ is $\bP$-ergodic.

Suppose $(\tW_n)_{n \geq 1}$ is $\bP$-ergodic.
Due to the stationarity of the Markov process $(\tW_n,Y_n)_{n \geq 1}$, for any $H \in \cF_\cY$ and any $G \in \cF_\cW$, we can define
\beqa
\eta(H|G) &\triangleq& \bprob{Y_n \in H | \tW_n \in G} \nonumber \\
 &\equiv& P(Y_1 \in H | X_1 \in \phi(G)) \label{eqn:proof:lemma:MarginalErgodicityImpliesJointErgodicity:aa}.
\eeqa
Then for $C,D \in \cF_\cY$ and $A,B \in \cF_\cW$, we have that
\beqa
&& \bprob{\tW_n \in B, Y_n \in D, \tW_1 \in A, Y_1 \in C} \nonumber \\
&=& \bprob{\tW_n \in B, \tW_1 \in A} \nonumber \\
&& \quad \times \; \bprob{Y_1 \in C | \tW_n \in B, \tW_1 \in A} \nonumber \\
&& \quad \times \; \bprob{Y_n \in D | Y_1 \in C, \tW_n \in B, \tW_1 \in A} \nonumber \\
&=&  \bprob{\tW_n \in B, \tW_1 \in A} \eta (C| A) \eta(D|B)  \label{eqn:proof:lemma:MarginalErgodicityImpliesJointErgodicity:a} \\
&\to& \bnu(A)\bnu(B)	\eta(C|A) \eta(D|B) \label{eqn:proof:lemma:MarginalErgodicityImpliesJointErgodicity:b}
\eeqa
where \eqref{eqn:proof:lemma:MarginalErgodicityImpliesJointErgodicity:a} follows from the memoryless, time-invariant nature of the channel
\eqref{eqn:defn:DMC} as well as \eqref{eqn:proof:lemma:MarginalErgodicityImpliesJointErgodicity:aa};
\eqref{eqn:proof:lemma:MarginalErgodicityImpliesJointErgodicity:b} follows from the
assumption that $(\tW_n)_{n \geq 1}$ is $\bP$-ergodic and Lemma~\ref{lemma:equivalentConditions:MarkovStationaryErgodic} (ii).
We note that $\bnu(A) \eta(C|A) = \bprob{W_1 \in A, Y_1 \in C}$ is the stationary distribution of the Markov chain $(\tW_n,Y_n)_{n\geq 1}$.  Thus
we have from \eqref{eqn:proof:lemma:MarginalErgodicityImpliesJointErgodicity:b} that the stationary Markov process $(\tW_n,Y_n)_{n\geq 1}$ is $\bP$-mixing (in the ergodic theory sense \cite[Sec 2.5]{bradley2005basic}); from Lemma~\ref{corollary:mixingImpliesErgodicity}, it is $\bP$-ergodic.
\end{proof}

\newcommand{\tprob}[1]{\tilde{\mathbb{P}} \parenth{#1}}
\newcommand{\tilP}{\tilde{P}}
\newcommand{\tE}{\tilde{\mathbb{E}}}
\newcommand{\dPdtP}{\frac{dP}{d \tilP }}

\newcommand{\Iueps}{I^{u,\epsilon}}
\newcommand{\IWeps}{I^{W,\epsilon}}
\newcommand{\Anueps}{A^{u,\epsilon}_n}
\newcommand{\AnWeps}{A^{W,\epsilon}_n}
\newcommand{\AnWOeps}{A^{\tilde{W}_{n+1},\epsilon}_n}
\newcommand{\epsin}{\epsilon_{i,n}}
\newcommand{\cFYip}{\cF^Y_{1:i-1}}
\newcommand{\KLi}{\kldist{\bpi_i}{\bpi_{i-1}}}
\newcommand{\KLin}{\kldist{\bpi_{i|n}}{\bpi_{i-1|n}}}
\newcommand{\SYinv}[1]{S^{-1}_{Y_{#1}}}
\newcommand{\SY}[1]{S_{Y_{#1}}}
\newcommand{\LLRpiIntrvl}[2]{\log \frac{\bpi_{#1}\parenth{#2}} { \bpi_{#1-1}\parenth{#2}}}
\newcommand{\LLRpiIntrvlAni} {\LLRpiIntrvl{i}{A_n}}
\newcommand{\LLRpiIntrvlBi} {\LLRpiIntrvl{i}{B}}
\newcommand{\tAin}{\tA_{i,n}}

We now establish a limiting property of the normalized information density:
\begin{lemma}
If $(\tW_{n})_{n\geq 1}$ is $\bP$-ergodic, then $\limn i_n(W) = C$.
\label{lemma:inf_density:capacity}
\end{lemma}
\begin{proof}
Suppose $(\tW_n)_{n \geq 1}$ is $\bP$-ergodic.  Then from Lemma~\ref{lemma:MarginalErgodicityImpliesJointErgodicity}, we have that
$(\tW_n,Y_n)_{n \geq 1}$ is $\bP$-ergodic.   Thus
\begin{eqnarray}
\lim_{n \to \infty} i_n(W) &=& \bE[ i(\tW_1,Y_1) | \tail_{\tW,Y}] \quad \bP-a.s. \label{eqn:proof:thm:achievingCapacityErgodicity:a} \\
                           &=& \bE[ i(\tW_1,Y_1)] \label{eqn:proof:thm:achievingCapacityErgodicity:b}  \\
                           &=& I(\tW_1;Y_1) \label{eqn:proof:thm:achievingCapacityErgodicity:c} \\
                           &=& C.\label{eqn:proof:thm:achievingCapacityErgodicity:d}
\end{eqnarray}
where \eqref{eqn:proof:thm:achievingCapacityErgodicity:a} follows from Theorem~\ref{theorem:PM:properties};
\eqref{eqn:proof:thm:achievingCapacityErgodicity:b} follows from the assumption that $(\tW_n)_{n \geq 1}$ is $\bP$-ergodic and
Lemma~\ref{lemma:MarginalErgodicityImpliesJointErgodicity}, implying that $(\tW_n,Y_n)_{n \geq 1}$ is $\bP$-ergodic;
\eqref{eqn:proof:thm:achievingCapacityErgodicity:c} follows from \eqref{eq:defn:mutualinformation:conditionalKLdivergence:a};
and \eqref{eqn:proof:thm:achievingCapacityErgodicity:d} follows from Theorem~\ref{theorem:PM:properties}.
\end{proof}

\subsection{Pulled Back Intervals}
\newcommand{\Deps}[2]{D^\epsilon_{#1:#2}}
\newcommand{\Dn}{D^{\epsilon}_{n+1}}
\newcommand{\Dij}{\Deps{i}{j}}

\newcommand{\An}{A^\epsilon_n}

We now define sets $\Dn$ as follows:
\begin{definition}\label{defn:Dn}
  Define $\tau^\epsilon: \cW \to \cF_{\cW}$ such that $\tau^\epsilon(u)$ is a convex open set, $\tau^\epsilon(u) \ni u$ for any $u \in \cW$, and $\bnu(\tau^\epsilon(u))=1-\epsilon$.
  Define $\Dn \in \cF_\cW$ as 
  \beqa 
  \Dn \triangleq \tau^\epsilon(\tW_{n+1}) \label{eqn:defn:Dn}
  \eeqa
\end{definition}
 For example, if $\cW=(0,1)$ then define
  \[\tau^\epsilon(u) = \begin{cases} (0,1-\epsilon) & u \in (0,1-\epsilon) \\ (\epsilon,1) & u \in (\epsilon,1) \end{cases}.\]
  The extensions to arbitrary dimension follow naturally.

  We now define the pulled-back intervals $\An$ which will serve as the open, convex sets in $\cW$ underlying the notion of achieving a rate in Definition~\ref{defn:achievability:rateR}.
\begin{definition}[Pulled-Back Intervals]
  \label{defintion:pulled_back_intervals}
Define the pulled back intervals as 
\beqa
\An \equiv A_{1,n} \triangleq S_{Y_{1:n}}^{-1}(\Dn) \label{eqn:defn:An}
\eeqa
\end{definition}

Figure~\ref{fig:AchieveRate} gives an example of a pulled back interval.  Notice that in this example, $\Dn$ has length $1-\epsilon$ and contains $\tW_{n+1}$, while $\An$ has length nearly zero and contains $W \equiv \tW_1$.
\begin{figure}[htbp]
  \centering
  \begin{overpic}[width=0.7\columnwidth]{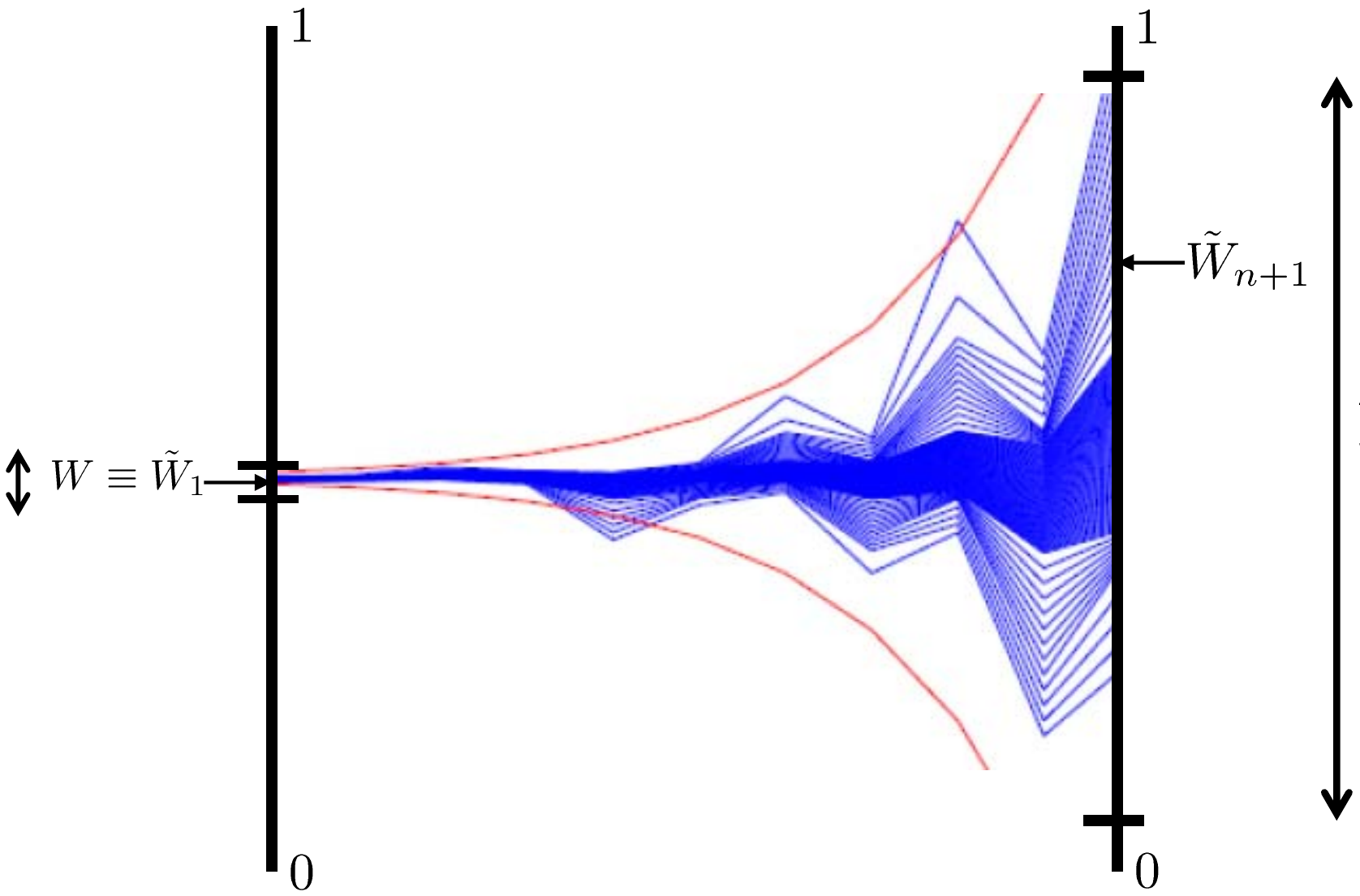}
   \put(-9,29){\small $\An$}
   \put(84,6){\tiny $\frac{1}{2}\epsilon$}
  \put(83,57){\tiny $1-\frac{1}{2}\epsilon$}
  \put(100,28.5){\small $\Dn$}
  \end{overpic}
  \caption{Mapping from the boundary points $\parenth{\frac{1}{2}\epsilon,1-\frac{1}{2}\epsilon}$ of $\Dn$ on the $\tW_{n+1}$ axis to the boundary points $\parenth{T_n(\frac{1}{2}\epsilon),T_n(1-\frac{1}{2}\epsilon)}$ of $\An$ on the $\tW_1$ axis. }
  \label{fig:AchieveRate}
\end{figure}

\newcommand{\Ain}{A_{i,n}^\epsilon}

\begin{lemma}\label{lemma:rate_sequence}
Defining 
\beqa
R_n &\triangleq& \frac{1}{n} \log \frac{\bpost_n(\An)}{\bpost_0(\An)} \label{eqn:lemma:rate_sequence:a}\\
\Ain &\triangleq& S_{Y_{1:i-1}}(\An)  \label{eqn:lemma:rate_sequence:b} \\
Z_{i,n} &\triangleq& \log \frac{P_{W|Y}(\Ain|Y_i)}{P_W(\Ain)} \label{eqn:lemma:rate_sequence:c}
\eeqa
the rate sequence $R_n$ can be equivalently represented as
\begin{align}
  R_n = \frac{1}{n} \sum_{i=1}^n Z_{i,n} \label{eqn:defn:Rn}.
\end{align}
\end{lemma}
\begin{proof}
Note that clearly from \eqref{eqn:lemma:rate_sequence:a}:
\beqa
R_n = \frac{1}{n} \sum_{i=1}^n \log \frac{\bpost_i(\An)}{\bpost_{i-1}(\An)}.
\eeqa
For any $A \in \cF_{\cW}$, and any $1 \leq i \leq n$,  note that:
\beqa
 \frac{\bpost_i(A)}{\bpost_{i-1}(A)} &=& \frac{\bprob{\tW_1 \in A | \cFYi}}{\bprob{\tW_1 \in A | \cFYip}} \nonumber \\
  &=& \frac{\bprob{\tW_i \in S_{Y_{1:i-1}} (A) | \cFYi}}{\bprob{\tW_i \in S_{Y_{1:i-1}} (A) | \cFYip}} \label{eqn:proof:rate_sequence:a} \\ 
  &=&  \frac{P_{\tW_i|Y_i}\parenth{S_{Y_{1:i-1}} (A)|Y_i}}{P_{\tW_i}\parenth{S_{Y_{1:i-1}} (A)}}   \label{eqn:proof:rate_sequence:b} \\
  &=& \frac{P_{\tW|Y}\parenth{S_{Y_{1:i-1}} (A)|Y_i}}{P_{\tW}\parenth{S_{Y_{1:i-1}} (A)}}   \label{eqn:proof:rate_sequence:c}
\eeqa
where \eqref{eqn:proof:rate_sequence:a} follows 
from \eqref{eqn:WinvertibleWithWn}; 
\eqref{eqn:proof:rate_sequence:b} follows from Lemma~\ref{lemma:PMscheme:independence}; and 
\eqref{eqn:proof:rate_sequence:c} follows from Theorem~\ref{theorem:PM:properties}.
Letting $A \equiv \Ain$ in \eqref{eqn:proof:rate_sequence:c} and exploiting definitions \eqref{eqn:lemma:rate_sequence:b} and \eqref{eqn:lemma:rate_sequence:c}, we have that
\beqas
\log \frac{\bpost_i(\An)}{\bpost_{i-1}(\An)} &=& \log \frac{P_{\tW|Y} \parenth{S_{Y_{1:i-1}} (\An)|Y_i}}{P_{\tW}\parenth{S_{Y_{1:i-1}}(\An)}}  \\
&=& \log \frac{P_{\tW|Y}(\Ain|Y_i)}{P_{\tW}(\Ain)} \\
&=& Z_{i,n}.
\eeqas

\end{proof}

We can now relate $Z_{1,n}$ to the information density $i(\tW_1,Y_1)$ as follows.
\begin{lemma}\label{lemma:probabilityRatioInformationDensity:new}
Define $\cG_{n+1} \subset \sigma(\tW_{n+1})$ to be the $\sigma$-algebra with atoms $(\Dn,{\Dn}^c)$ and 
$\cH_{n} \triangleq \cG_{n+1} \otimes \cFYn$.  Then: %
\beqa
Z_{1,n} = -\log \bE \brackets{e^{-i(\tW_1,Y_1)} | \cH_n}.
\eeqa
\end{lemma}
\begin{proof}
\renewcommand{\dPdtP}{\frac{d \bP}{d\tP}}
Define $\tP$ as the probability measure on $(\Omega,\cF)$ for the random process $(\tW_n,Y_n)_{n \geq 1}$ such that
\beqa
\dPdtP =\frac{dP_{\tW_1,Y_1}}{d (P_{\tW_1} \times P_{Y_1})} =  e^{i(\tW_1,Y_1)}. \label{eqn:defn:dPdtP:a}
\eeqa
where the latter equality in \eqref{eqn:defn:dPdtP:a} follows from \eqref{eqn:defn:informationdensity:b}.
Since by definition, $\bP$ and $\tP$ only differ in their distributions on $(\tW_1,Y_1)$, we have that for any $n \geq 1$:
\beqa
\dPdtP = \frac{d P_{\tW_{1},Y_{1}}}{d\tilP_{\tW_{1},Y_{1}}} = \frac{d P_{\tW_{1},Y_{1:n}}}{d\tilP_{\tW_{1},Y_{1:n}}} = \frac{d P_{\tW_{n+1},Y_{1:n}}}{d\tilP_{\tW_{n+1},Y_{1:n}}} \label{eqn:proof:lemma:probabilityRatioInformationDensity:new:intermediate}
\eeqa
where the final equality in \eqref{eqn:proof:lemma:probabilityRatioInformationDensity:new:intermediate} follows from 
\eqref{eqn:WinvertibleWithWn}.
As such,
\beqa
Z_{1,n} &=& \log \frac{P_{\tW_1,Y_{1}}(A_{n},dY_{1})}{\tilP_{\tW_1,Y_{1}}(A_{n},dY_{1})} \nonumber\\
&=& \log \frac{P_{\tW_1,Y_{1:n}}(A_{n},dY_{1:n})}{\tilP_{\tW_1,Y_{1:n}}(A_{n},dY_{1:n})} \nonumber \\
&=& \log \frac{P_{\tW_{n+1},Y_{1:n}}(\Dn,dY_{1:n})}{\tilP_{\tW_{n+1},Y_{1:n}}(\Dn,dY_{1:n})} \label{eqn:proof:lemma:probabilityRatioInformationDensity:new:b}\\
                            &=& \log \frac{d P_{\tW_{n+1},Y_{1:n}|\cH_n}}{d\tilP_{\tW_{n+1},Y_{1:n}|\cH_n}}  \label{eqn:proof:lemma:probabilityRatioInformationDensity:new:c}\\
                            &=& \log \frac{d \mathbb{P}_{|\cH_n}}{d \tilde{\mathbb{P}}_{| \cH_n}  }  \label{eqn:proof:lemma:probabilityRatioInformationDensity:new:cc}\\
                          &=& \log \tilde{\E}  \brackets{ \dPdtP  \big| \cH_n}  \label{eqn:proof:lemma:probabilityRatioInformationDensity:new:d} \\
                            &=& \log \tilde{\E} \brackets{ e^{i(\tW_1,Y_1)} \big| \cH_n}  \label{eqn:proof:lemma:probabilityRatioInformationDensity:new:e} \\
   &=& -\log \bE \brackets{e^{-i(\tW_1,Y_1)} \big| \cH_n} \label{eqn:proof:lemma:probabilityRatioInformationDensity:new:f}
\eeqa
where \eqref{eqn:proof:lemma:probabilityRatioInformationDensity:new:b} follows from Definition~\ref{defintion:pulled_back_intervals} and \eqref{eqn:WinvertibleWithWn};
\eqref{eqn:proof:lemma:probabilityRatioInformationDensity:new:c}  follows from the definition of $\cH_n$ in this Lemma;
\eqref{eqn:proof:lemma:probabilityRatioInformationDensity:new:cc}  follows from \eqref{eqn:proof:lemma:probabilityRatioInformationDensity:new:intermediate};
\eqref{eqn:proof:lemma:probabilityRatioInformationDensity:new:d} follows from classical probability theory \cite[Lemma 5.2.4]{gray1990eai}: the Radon-Nikodym derivative for a restriction is  the conditional
expectation of the original Radon-Nikodym derivative;
 \eqref{eqn:proof:lemma:probabilityRatioInformationDensity:new:e} follows from \eqref{eqn:proof:lemma:probabilityRatioInformationDensity:new:intermediate}; 
and \eqref{eqn:proof:lemma:probabilityRatioInformationDensity:new:f} follows for Bayes' rule: $\bE[V|\cH]=\frac{\tilde{\E}[V\dPdtP|\cH]}{\tilde{\E}[\dPdtP|\cH]}$ with $V=\parenth{\dPdtP}^{-1}=e^{-i(\tW_1,Y_1)}$.
\end{proof}

\begin{lemma} \label{lemma:AchievingCapacity:condExpectation}
If the PM scheme is reliable, then
\beqa
\lim_{n \to \infty} \bE[Z_{1,n} | \cFYn] =i(\tW_1,Y_1) \quad \bP-a.s. \label{eqn:lemma:AchievingCapacity:condExpectation}
\eeqa
\end{lemma}
\begin{proof}
Since $\Dn \in \cG_{n+1}$, with $\An=T_{1:n}(\Dn)$, we have from Lemma~\ref{lemma:probabilityRatioInformationDensity:new} and Jensen's inequality that
\begin{align}
\bE[Z_{1,n} | \cFYn] &=  \bE \brackets{-\log\bE \brackets{e^{-i(\tW_1,Y_1)} \big| \cH_n} \Big| \cFYn} \nonumber\\
                                &\geq -\log \bE \brackets{ \bE \brackets{e^{-i(\tW_1,Y_1)} \big| \cH_n} \Big|\cFYn} \nonumber\\
                                &= -\log \bE \brackets{e^{-i(\tW_1,Y_1)} \big| \cFYn} \nonumber
\end{align}
Therefore, taking a $\liminf$:
\begin{align}
\liminf_{n \to \infty}  \bE[Z_{1,n} | \cFYn] &\geq \liminf_{n \to \infty} 
 -\log \bE \brackets{e^{-i(\tW_1,Y_1)} \big| \cFYn} \nonumber \\
  &= -\log \bE \brackets{e^{-i(\tW_1,Y_1)} \big| \cFYinf} \nonumber \\
  &= -\log \bE \brackets{e^{-i(\tW_1,Y_1)} \big| \sigma(\tW_1) \vee \cFYinf} \label{eqn:lemma:AchievingCapacity:condExpectation:a}\\
  &= -\log \parenth{e^{-i(\tW_1,Y_1)}} \nonumber\\
  &= i(\tW_1,Y_1) \nonumber \quad \bP-a.s.
\end{align}
where \eqref{eqn:lemma:AchievingCapacity:condExpectation:a} follows from the assumption that the PM scheme is reliable.

Now since $e^{-u}$ is convex, we have from Jensen's inequality that
$\bE \brackets{e^{-i(\tW_1,Y_1)} \big| \cG_n} \geq e^{-\bE \brackets{ i(\tW_1,Y_1)| \cG_n}}$.
Thus we have that $Z_{1,n}=-\log \bE \brackets{e^{-i(\tW_1,Y_1)} \big| \cG_n} \leq \bE \brackets{ i(\tW_1,Y_1)| \cG_n}$. Thus from the tower law of conditional expectation:
\begin{align}
\limsup_{n \to \infty} \bE[Z_{1,n} | \cFYn]  &\leq& \limsup_{n \to \infty} \bE[i(\tW_1,Y_1) | \cFYn] \nonumber\\
                      &=&  \bE[i(\tW_1,Y_1) | \cFYinf] \nonumber\\
                           &=&  \bE[i(\tW_1,Y_1) | \sigma(\tW_1) \vee \cFYinf] \label{eqn:lemma:AchievingCapacity:condExpectation:b}\\
                           &=& i(\tW_1,Y_1) \quad \bP-a.s. \nonumber
\end{align}
where \eqref{eqn:lemma:AchievingCapacity:condExpectation:b} follows from the assumption that the PM scheme is reliable.
\end{proof}

\newcommand{\Ainlong}{S_{Y_i}^{-1} \circ \cdots \circ S_{Y_n}^{-1} (\Dn)}
\newcommand{\Dnshifted}{D_{n-i+2}^\epsilon}
\newcommand{\shiftedAin}{S_{Y_1}^{-1} \circ \cdots \circ S_{Y_{n-i+1}}^{-1} (\Dnshifted)}

\begin{theorem} \label{thm:achievingCapacityErgodicity}
  For the PM scheme, if $(\tW_n)_{n \geq 1}$ is $\bP$-ergodic, then any rate $R <C$ is achievable.
\end{theorem}
\begin{proof}
From Assumption~\ref{assump:finiteCapacity}, $\bE[i(\tW_1,Y_1)]=C < \infty$; thus $(Z_{1,n}: n\geq 1)$ and $(\bE[Z_{1,n}|\cFYn])_{n \geq 1}$  are uniformly integrable \cite[Lemma 5.4.1]{gray1990eai}.  From \eqref{eqn:lemma:AchievingCapacity:condExpectation}, it follows from $L^1$ convergence that $\bE[Z_{1,n}] \to \bE[i(\tW_1,Y_1)]=C$. 
Combining \eqref{eqn:defn:An} with \eqref{eqn:lemma:rate_sequence:b}, it follows that
 $\Ain= S_{Y_i}^{-1} \circ \cdots \circ S_{Y_n}^{-1} (\Dn)$.  In addition, from Lemma~\ref{lemma:PMscheme:independence}, $(\tW_n,Y_n)$ is stationary.  As such, we have that for any $i \leq n$,
\beqa
&& \bE[Z_{i,n}] \nonumber \\
 &=& \bE \brackets{\log \frac{P_{W|Y}(\Ain|Y_i)}{P_W(\Ain)} } \nonumber \\
             &=& \bE \brackets{\log \frac{P_{W|Y} \parenth{\Ainlong|Y_i}}{P_W \parenth{\Ainlong}}}\nonumber \\
             &=& \bE \brackets{\log \frac{P_{W|Y} \parenth{\shiftedAin|Y_1}}{P_W \parenth{\shiftedAin}}}\nonumber \\
             &=& \bE[Z_{1,n-i+1}]
\eeqa
Thus it follows that
\beqas
\bE[R_n] = \frac{1}{n}\sum_{i=1}^n \bE[Z_{i,n}] = \frac{1}{n}\sum_{i=1}^n \bE[Z_{1,n-i+1}] 
\eeqas
And so from Ces\`aro, since $\bE[Z_{1,n}] \to C$, it follows that $\bE[R_n] \to C$.
  
Now that we have established the necessary properties of the pulled-back intervals and the rate sequence in expectation, we can showcase analogous results in the $\bP$-a.s. sense. %
Consider without  loss of generality:
\[\Omega = \braces{ \omega = \parenth{\tw_1,y_1,\tw_2,y_2,\ldots} \in \reals^{2 \times \infty}}. \]
Define $\cF$ to be the Borel sigma algebra for $\Omega$. For any $\omega = \parenth{\tw_1,y_1,\tw_2,y_2,\ldots}$,  define the \textit{shift operator} $\Upsilon$ to be given by $\Upsilon(\omega) \triangleq \parenth{\tw_2,y_2,\tw_3,y_3,\ldots}$.  Since $(\tW_n,Y_n)_{n\geq 1}$ is stationary,  $\Upsilon$ is a $\bP$-measure-preserving transformation \cite[Prop. 6.11]{breiman1992probability}.  For any event $A \in \cF$, we say that $A$ is \textit{invariant} if $A=\Upsilon(A)$.  Any invariant set has the property that it lies in the tail  $\tail_{\tW,Y}$ where
\[ \tail_{\tW,Y} = \bigcap_{n \geq 1} \sigma(\tW_n,Y_n, \tW_{n+1},Y_{n+1},\ldots). \]

Consider any $v \geq 0$. Define the following event
\begin{align}
  A_{v} &\triangleq \braces{ \omega: \liminf_{n \to \infty} R_n(\omega) \leq v } \nonumber \\
    &= \braces{\omega: 
\frac{1}{n} \sum_{i=1}^{n} \log \frac{P_{W|Y}(\Ain|Y_i)}{P_W(\Ain)}(\omega)
  \leq v} \label{eqn:proof:capacitystuff:AA}
\end{align}
where \eqref{eqn:proof:capacitystuff:AA} follows from \eqref{eqn:lemma:rate_sequence:c}
and \eqref{eqn:defn:Rn}. Then note that 
\beqas
A_{v} &=& \braces{ \omega: \liminf_{n \to \infty} \frac{n+1}{n} R_{n+1}(\omega) \leq v }   \\
      &=& \Big\{ \omega: \liminf_{n \to \infty} \frac{1}{n} Z_{1,n+1} (\omega)   \\
      &&   \quad \quad\quad\quad +  \frac{1}{n} \sum_{i=2}^{n+1} \log \frac{P_{W|Y}(\Ain|Y_i)}{P_W(\Ain)}       \leq v \Big\} \\
      &=& \braces{\omega: \liminf_{n \to \infty} \frac{1}{n} \sum_{i=2}^{n+1} \log \frac{P_{W|Y}(\Ain|Y_i)}{P_W(\Ain)}       \leq v}
       \\
      &=& \Upsilon(A_v)
\eeqas

As such, from \cite[Prop 6.17]{breiman1992probability}, $A_v \in \tail_{\tW,Y}$.  Thus, from ergodicity,
$\bprob{A_v}=0$ or $\bprob{A_v}=1$.
Suppose that $1=\bprob{A_v}=\bprob{\liminf_{n \to \infty} R_n \leq v}$. Then
\beqa
\liminf_{n \to \infty} \bprob{R_n \leq v} \geq \bP(\liminf_{n \to \infty} R_n \leq v) = 1.
\eeqa
Letting $v=C-\delta$, then we have that
\beqa
\limsup_{n\to\infty} \bprob{R_n > C-\delta} &=& 1-\liminf_{n \to \infty} \bprob{R_n \leq C-\delta} \nonumber \\
                                             &=& 0. \label{eqn:limsuptozero}
\eeqa
Thus, since $0 \leq \bprob{R_n > u} \leq 1$ for any $u$, we have from the reverse Fatou's lemma that
\beqa
\limsup_{n \to \infty} \bE \brackets{R_n}
&=&  \limsup_{n \to \infty} \int_{u=0}^\infty \bprob{R_n > u} du \nonumber \\
&\leq&   \int_{u=0}^\infty \limsup_{n \to \infty} \bprob{R_n > u} du \nonumber \\
&=&   \int_{u=0}^{C-\delta} \lim_{n \to \infty} \bprob{R_n > u} du \label{eqn:limsup:upperbound} \\
&\leq& C-\delta \nonumber \\
&<& C. \nonumber 
\eeqa
where \eqref{eqn:limsup:upperbound} follows from \eqref{eqn:limsuptozero}.
But that is a contradiction because we have already shown that $\lim_{n \to \infty} \bE \brackets{R_n} =C$. Thus it must be that $\bprob{\liminf_{n \to \infty} R_n  \leq C-\delta}=0$, or equivalently that $\bprob{\liminf_{n \to \infty} R_n  > C-\delta}=1$.  This holds for any $\delta > 0$ and so as $\delta \to 0$, we have that
\[\bprob{\liminf_{n \to \infty} R_n \geq C}=1. \]

\end{proof}

We can now summarize all of the above in the main theorem of the paper:
\begin{theorem}\label{theorem:equivalenceThreeConditionsPM:reliable:ergodic:capacity}
The following conditions are equivalent:
\benum
\item The PM scheme is reliable.
\item $(\tW_n)_{n \geq 1}$ is $\bP$-ergodic.
\item The PM scheme achieves any rate below capacity.
\eenum
\end{theorem}
\begin{proof}
$(1) \Leftrightarrow (2)$: Theorem~\ref{thm:necessarySufficientConditionsPM:reliability};
$(2) \Rightarrow (3)$: Theorem~\ref{thm:achievingCapacityErgodicity};
$(3) \Rightarrow (1)$: Definition~\ref{defn:achievability:rateR} property (ii).
\end{proof}

\section{Applications and Examples}\label{sec:applications}

In this section, we utilize OTT to demonstrate the construction of a nonlinear
diffeomorphic map $S$ for PM schemes in arbitrarily high dimensions.

\begin{figure*}[t]
  \centering
  \begin{overpic}[width=\textwidth]{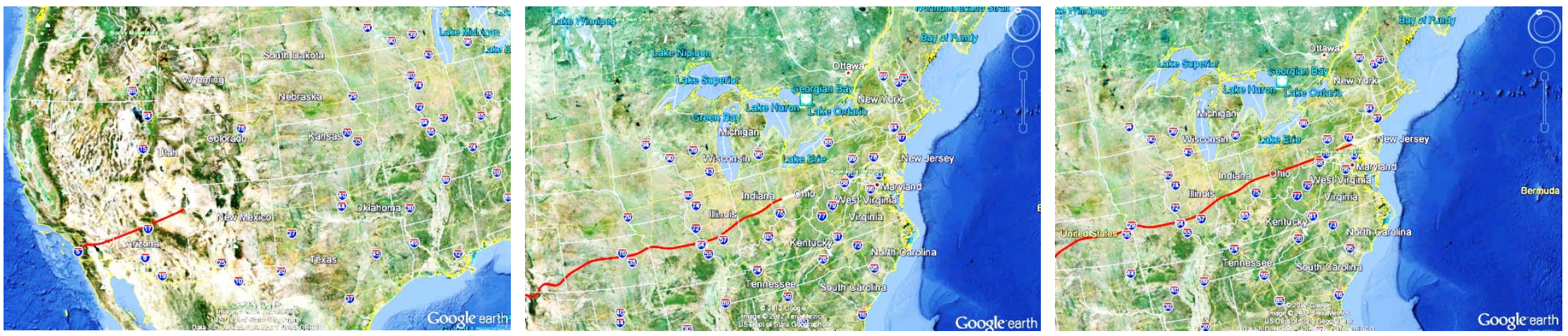}
  \put(1,-1){\tiny(a)}
  \put(34,-1){\tiny(b)}
  \put(68,-1){\tiny(c)}
  \end{overpic}
  \caption{Result of a brain-computer interface experiment. A smooth path (red)
  is specified sequentially using the PM scheme within the context of a
  brain-machine interface for the BSC shown in Example \ref{eg:BscPath}.}
  \label{fig:googleearth}
\end{figure*}

\subsection{One-Dimensional Posterior Matching Schemes}

\begin{example}[Original PM Scheme] \label{invCDF}
    When $\cW=(0,1)$ we now show how we can recover the posterior matching
    scheme by Shayevitz \& Feder in \cite{Shayevitz2011pm} using optimal transport theory.
    It is well-known \cite{villani2009optimal} if $c(w,v)=h(u-v)$ for some strictly convex $h: \reals \to \reals$, then
     the optimal map for $\OTTcost(P,Q,c)$ where $P \in \probSimplex{(0,1)}$ and $Q \in \probSimplex{(0,1)}$ is uniform $(0,1)$, $S^*$ is the map $S^*(u)=F(u)$ where $F$ is the CDF associated with $P$.  Within the case of the PM scheme, $P \equiv P_{W|Y=y}$ and $Q = P_W$, the uniform distribution on $(0,1)$.  As such, $S_y^*(u)=F_{W|Y_1=y}(u)$ which is exactly the PM scheme from \cite{Shayevitz2011pm}.
\end{example}

    \begin{remark}
      Note that clearly, $S_y(u) = 1-F_{W|Y_1=y}(u)$ is also a posterior matching scheme. This means that there
      are many maps that induce a posterior matching scheme. This
      special case was also discussed in \cite[Example 3.2.14]{rachev1998mass}.
    \end{remark}

\begin{example}[Horstein]
  As an example, for $\cW=[0,1]$, for a binary symmetric channel (BSC) with
  $\cY=\cX=\braces{0,1}$ and crossover probability $p$,
  \beqa
  \bP(Y=y|X=x) &=& \begin{cases} p, \quad y \neq x \nonumber\\ 1-p, \quad y=x
          \end{cases}
  \eeqa
  we have that the Horstein scheme \cite{horstein1963stu}:
  \begin{subequations} \label{eqn:HorsteinScheme}
  \begin{align}
  X_{n+1}&=\phi(\tW_{n+1})=\left\{
            \begin{array}{lll}
              0, & \tW_{n+1}\in[0,\frac{1}{2}] \\ 1, & \tW_{n+1}
              \in(\frac{1}{2},1]
            \end{array}
          \right.  \\
     &= \begin{array}{lll} 0, & W < m_n \equiv \text{median} \parenth{ \bpost_n}
              \\ 1, & W \geq m_n \equiv \text{median} \parenth{ \bpost_n}
            \end{array}
  \end{align}
  \end{subequations}
  is a posterior matching scheme that achieves capacity (see \cite{Shayevitz2011pm}). %
\end{example}

\begin{example}[BSC and Brain-Computer Interfaces]\label{eg:BscPath}
  This algorithm was originally implemented in \cite{omar2010feedback} and was
  used to specify a smooth path that is in one-to-one correspondence (via
  arithmetic coding) with a point $W \in (0,1)$.  The computer takes
  observations $Y_{1:n}$ to compute the posterior $\bpost_n$ and specifies a query
  point $m_n$ to the human as the median of $\bpost_n$.  The human specifies
  $X_n=0$ or $X_n=1$ in response to a query point, as given by
  \eqref{eqn:HorsteinScheme}.   The human user of a brain-computer interface can
  utilize EEG motor imagery to provide a series of binary inputs, imagine left
  ($X_{n+1}=0$) or imagine right ($X_{n+1}=1$). Assuming that the channel from
  human brain to EEG measurements is input-symmetric, we can model any EEG
  classification system as a binary symmetric channel (BSC) with outputs
  $Y_{n+1}=0$ or $Y_{n+1}=1$, and the process continues. The comparison between
  $W$ and $m_n$ performed by the human can be done visually by the decoder
  displaying $m_n$  as an ordered sequence and the user performing a
  lexicographic comparison with Cover's enumerative source coding \cite{cover1973enumerative}.  In this case, the first
  point where $W$ and $m_n$ deviate resulting in a counter-clockwise direction
  from $m_n$ to $W$ means that $W< m_n$, and vice versa for $W \geq m_n$.  This
  has the effect of the computer to sequentially attaining increasing confidence
  about longer sub-paths of the message $W$.  Figure~\ref{fig:googleearth}
  represents a simulation of this paradigm, overlaid on Google Earth, to
  demonstrate its feasibility of use in real-world scenarios.
\end{example}

\subsection{The Optimal Symmetric Brenier Map for Higher Dimensional Problems}\label{subsec:symbren}
We now provide an example for the Brenier map, which is the optimal solution to
$\OTTcost(P,Q, I_{d \times d})$.

We have previously demonstrated a BCI paradigm on a one-dimensional alphabet in
\cite{omar2010feedback}, and in Example~\ref{eg:BscPath}. To generalize,
consider a BCI with which the subject wishes to ``zoom in'' onto a point in 2D
space, e.g. a picture or a map \cite{tantiongloc2016information}. We treat this as a scenario where
$\dim(\cW)=\dim(\cX)=2$, and more specifically $\cW=[0,1]^2$, $\cX = \cY =
\braces{0,1}^2$. We consider a scenario where a BCI can allow for a human to
signal $X_n$ in one of four categories  \cite{schlogl2005characterization}. For
example, suppose the subject wants to zoom in onto Shannon's face in Figure
\ref{fig:shannon} (left panel, indicated by $\tW_1$) taken from \cite{shannonMouse}.
Suppose we record the subject's neural signals to extract his/her motor intent,
which is restricted to specifying one of the four quadrants of $\cW$ as
partitioned by the red cross in Figure~\ref{fig:shannon}:
\begin{align*}
  X_n&=\phi(\tW_n)\\
  &=\left\{
            \begin{array}{lll}
              (0,0), & \tW_n[1]\in(0,\frac{1}{2}], &\tW_n[2]\in(0,\frac{1}{2}] \\
              (0,1), & \tW_n[1]\in (0,\frac{1}{2}], &\tW_n[2]\in(\frac{1}{2},1) \\
              (1, 0), & \tW_n[1]\in(\frac{1}{2},1], &\tW_n[2]\in(0,\frac{1}{2}] \\
              (1, 1) & \tW_n[1]\in(\frac{1}{2},1], &\tW_n[2]\in(\frac{1}{2},1)
            \end{array}
          \right.
\end{align*}
The motor intent $x_n$ are input to a quadratic symmetric channel (QSC),
modeling $y$  as the output of a classifier based upon neural recordings, with
conditional probability:
\beqas
\bP(Y=y|X=x) &=& \begin{cases}
					\frac{p}{3}, \quad y \in \braces{0,1}^2 \setminus \braces{x} \nonumber\\
					1-p, \quad   y=x
				\end{cases}.
\eeqas
In this setting, the posterior distribution $P_{W|Y=y}$ is piece-wise constant
over the four quadrants $\braces{\phi^{-1}(x): x \in \braces{0,1}^2}$.  
Suppose that
$y_1=x_1$, give by the ordered pair $(0,1)$. By applying the Brenier map $S_{y_1}$ to the original Figure
(e.g. Figure~\ref{fig:shannon}, left panel), we have transformed it into the one
in Figure~\ref{fig:shannon}, right panel. It is clear that Shannon's face has
been enlarged and zoomed in onto in areas of higher posterior probability.
\begin{figure}[t]
  \centering
  \includegraphics[width=0.98\columnwidth]{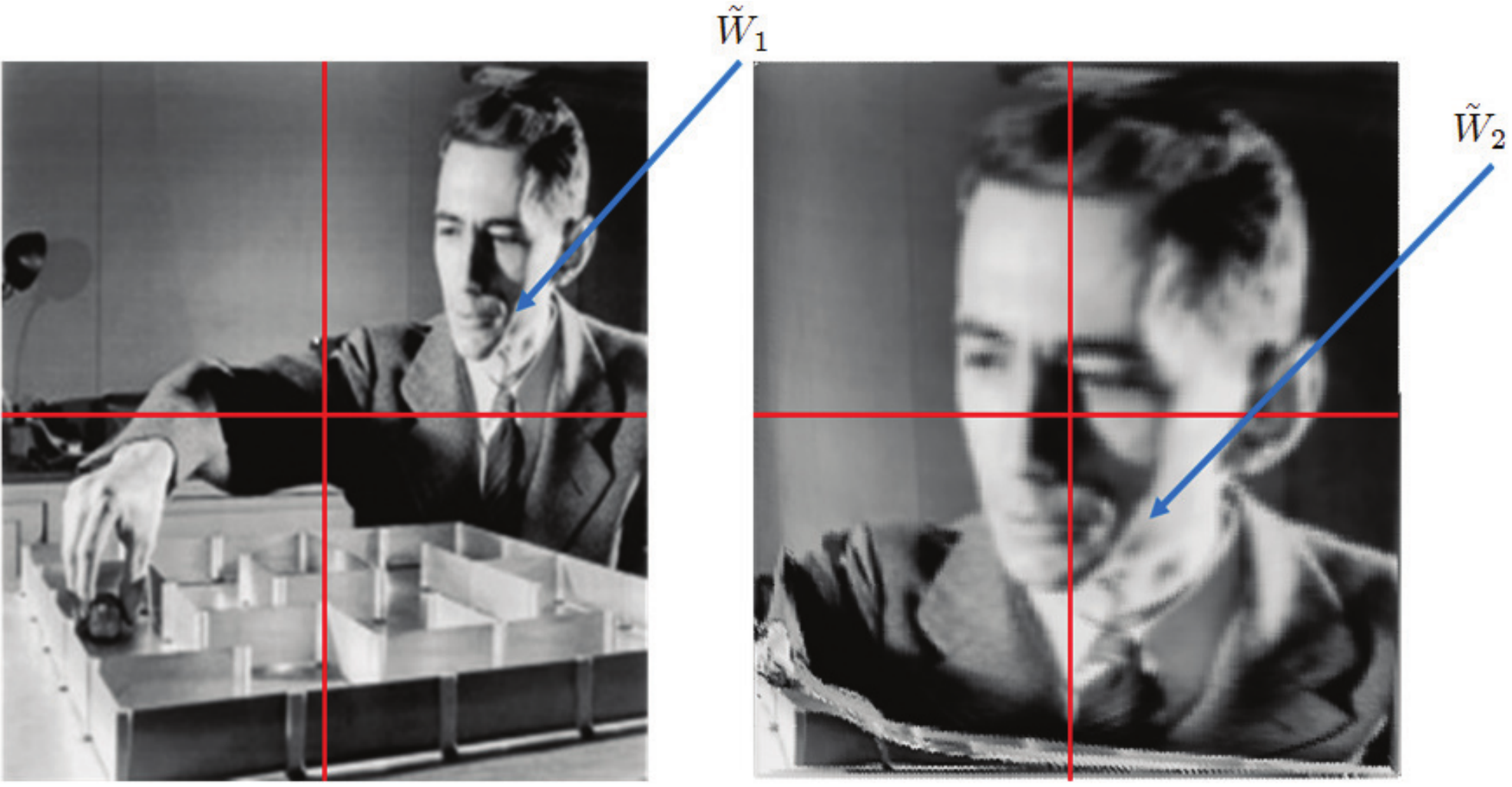}
  \caption{The Brenier optimal transport  posterior matching scheme, zooming in
  on a point in a picture of Claude Shannon and his mouse.   Left: Original picture of Shannon and his mouse from \cite{shannonMouse}.  Right: optimally computed `zoomed' in picture after application of $S_{y_1}$, where $y_1=1$ pertaining to the top right quadrant. }
  \label{fig:shannon}
\end{figure}

\subsection{The Knothe-Rosenblatt Map for Higher Dimensional Problems} \label{subsec:knhd}
Assume $\cW \subset \R^d$, and for any $w \in \cW$, the $k$-th component of $w$
is denoted by $w[k]$, and likewise for $v$. We now consider the 
problem $\OTTcost(P_{W|Y=y},P_W, c_M)$ for $M \succ 0$ to find a map $S_y$ for which $S_y \# P_{W|Y=y} = P_W$ in Corollary~\ref{corollary:unique:S}.  Assume that $M \equiv
M_\epsilon$ is given by
\beqa
M_\epsilon = \brackets{\begin{array}{c c c c}
               \alpha_\epsilon[1] & 0 & \ldots & 0 \\
               0 & \alpha_\epsilon[2] & \ldots  & 0 \\
               \ldots & \ldots & \ldots & \ldots \\
               0 & 0 & \ldots  & \alpha_\epsilon[d]
               \end{array}
               } \label{eqn:KnotheRosenblatt:Mepsilon}
\eeqa
\begin{lemma} \label{lemma:KnotheRosenblatt}
For $\cW=[0,1]^d$, suppose 
\beqa
\lim_{\epsilon \to 0} \frac{\alpha_\epsilon[k+1]}{\alpha_\epsilon[k]} = 0, \quad k=1,\ldots,d-1. \label{eqn:ratio:alpha}
 \eeqa
Then the optimal solution to $\OTTcost(P_{W|Y=y},P_W, c_{M_\epsilon})$ converges as $\epsilon \to 0$ to a Knothe-Rosenblatt map.  The limiting PM scheme becomes
\begin{subequations}\label{eqn:KR-PM}
\begin{align}
w_{n+1}[1] &= F_{W_n[1]|Y_n=y_n}(w_n[1]) \\
w_{n+1}[2] &= F_{W_n[2]|Y_n=y_n,W_{n+1}[1]=w_{n+1}[1]}(w_n[2]) \\
\ldots      & \ldots   \nonumber
\end{align}
\end{subequations}
\end{lemma}
\begin{proof}
  That the limit in \eqref{eqn:ratio:alpha} renders the Knothe-Rosenblatt map
  follows from \cite{santambrogio2010knothe}.  That the Knothe-Rosenblatt map for
  the scenario where $P\equiv P_{W|Y=y}$ and $Q \equiv P_W$ is given by \eqref{eqn:KR-PM} is a 
  direct application of \cite{rosenblatt1952remarks}.
\end{proof}
In this case, $M_\epsilon$ has increasing weights $\alpha_\epsilon[k]$ in $k$.   As such,
preference is given more to certain axes of the message point as compared to the
other. As such, the selection of the lopsided matrix $M_\epsilon$ in
$\OTTcost(P_{W|Y=y},P_W, M_\epsilon)$ provides a rational design methodology
to elicit PM schemes with unequal error protection \cite{borade2009unequal,nakiboglu2013bit},
\subsection{Multi-Antenna Communication as Multivariate Gaussian Channels} \label{sec:Gaussian}

In this section we will take two routes to solve for the optimal mapping in an
example case of a multi-antenna additive Gaussian noise with causal feedback.
We consider the case where the covariance matrices of the noise and the capacity
achieving inputs are not necessarily the identity matrix. We will find optimal
couplings with respect to the symmetric Brenier cost and the Knothe-Rosenblatt
cost, and show that in both cases, the schemes achieve capacity - although in
very different ways.  The former scheme will be shown to be the
multi-dimensional analogue to the `innovations' scheme by Schalkwijk and Kailath
\cite{schalkwijk1966csaOne}, while the latter can be interpreted from a
`onion-peeling' or `successive cancellation' perspective \cite{zhang2007successive}.

Suppose we have a multi-antenna communication problem with feedback
(Figure~\ref{fig:MIMO-feedback}).  As such, we model $\cX = \cY = \reals^d$,
where $d$ is the number of transmit antennas and also the number of receive
antennas.  Naturally, to develop a PM scheme where $\cX=\cW$ satisfying Assumption~\ref{assumption:cW}, it is desirable for
$\Phi: \cW \mapsto \cX$ to be an invertible map for which $\dim(\cW)=\dim(\cX)$.
As such, we model $\cW = [0,1]^d$. We model the Gaussian noise as $Z_i \sim
\mathscr{N}(0, \Sigma_N) \in \reals^d$ where $\Sigma_N$ is not necessarily
diagonal.  The received signal is given by $Y_i = X_i + Z_i$.

\begin{figure}[ht]
  \centering
  \begin{overpic}[width=0.7\columnwidth]{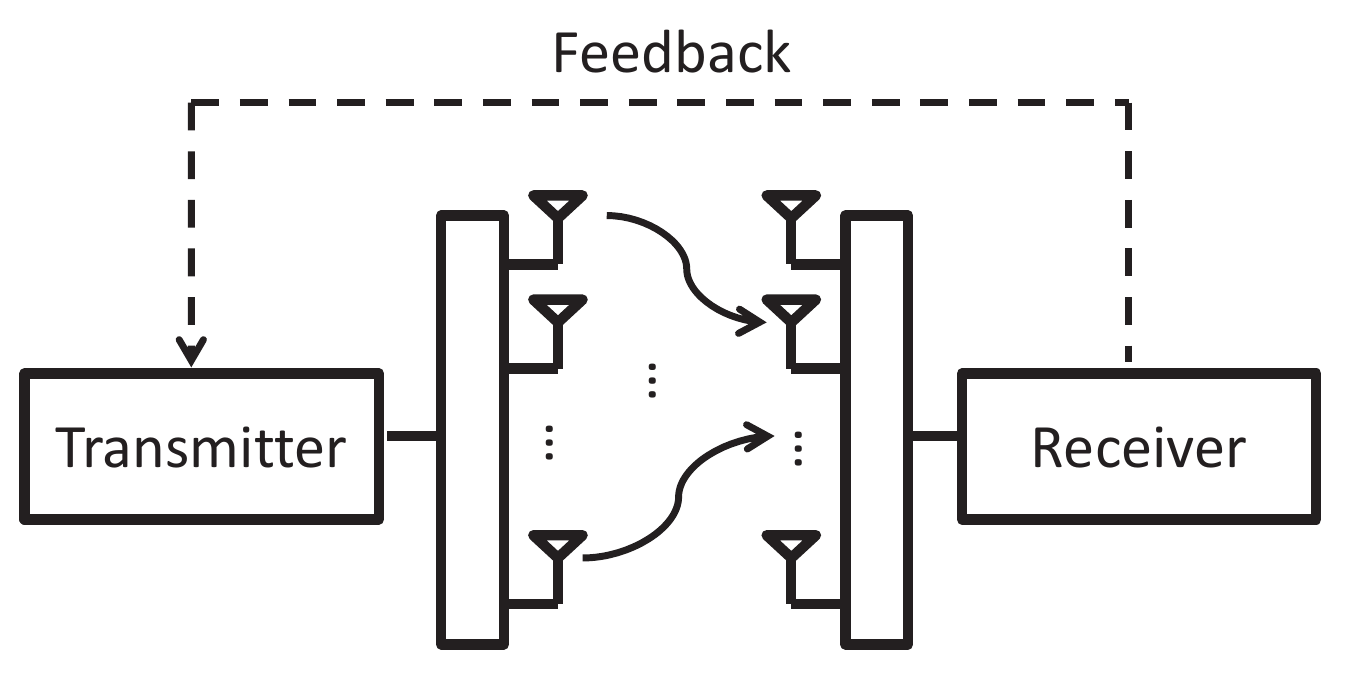}
  \end{overpic}
  \caption{Feedback  communication over a Gaussian MIMO channel with feedback.  $\dim(\cX) > 1$.}
  \label{fig:MIMO-feedback}
\end{figure}

Under an average power constraint, the optimal input distribution is given by $\px = \mathscr{N}(0,
\Sigma_X)$ for some $\Sigma_X$.  Note that $P_{X_1|Y_1} =\mathscr{N}(\E(X_1|Y_1),
\Sigma_{X|Y})$.  It can be directly shown \cite[Prop 3.4.4]{hajek2009rp} that
 $\E(X_1|Y_1) = \Sigma_X
\Sigma^{-1}_Y Y_1$ and $\Sigma_{X|Y} = \Sigma_X - \Sigma_X
\Sigma^{-1}_Y \Sigma_X$.
From Corollary~\ref{cor:PM-compatible-phi-invertible}, it follows that
it is our desire to construct a PM-compatible map $\bar{S}_y \deq \phi \circ S_y \phi^{-1}$ such that
$\bar{S}_y \# P_{X_1|Y_1=y} = P_{X}^*$.

\subsubsection{The Optimal Brenier Map for Higher Dimensional Problems}
Here we consider finding a diffeomorphism
$\bar{S}_{y}: \reals^d \to \reals^d$ that solves $\OTTcost(P_{X_1|Y_1=y},P_X^*, I_{d \times d})$.
\cite[Theorem 3.4.1]{rachev1998mass} provides the solution:
\begin{eqnarray}
  X_{n+1} &=& \bar{S}_{Y_n}(X_n) = \alpha (X_n - \Sigma_X \Sigma^{-1}_Y Y_n)
  \label{eqn:snr} \\
  \alpha &=& \Sigma^{\frac{1}{2}}_{X} \left( \Sigma^{\frac{1}{2}}_{X}
  \Sigma_{X|Y} \Sigma^{\frac{1}{2}}_{X} \right)^{-\frac{1}{2}}
  \Sigma^{\frac{1}{2}}_{X} \nonumber
\end{eqnarray}
Using MMSE estimation theory and Corollary~\ref{cor:PM-compatible-phi-invertible}, we can alternatively find the symmetric Brenier map by first positing that
\begin{eqnarray}
  X_{n+1} = \alpha(X_n - \beta Y_n). \label{eqn:pmguess}
\end{eqnarray}
with unknown $\alpha$ and $\beta$. $\beta = \Sigma_X
\Sigma^{-1}_Y$ ensures that $X_{n+1}$ is independent of $Y_n$ (by joint Gaussianity and MMSE estimation), so we only need
to find an invertible $\alpha$. Because all variables are jointly Gaussian, and
since both sides of \eqref{eqn:pmguess} have zero-means, we can simply operate
on covariance matrices.
\begin{equation*}
  \Sigma_X = \alpha(\Sigma_{X|Y})\alpha^T = \alpha(\Sigma_X - \Sigma_X \Sigma^{-1}_Y \Sigma_X)\alpha^T  %
\end{equation*}
It can be verified that this leads to the same linear algebra problem
encountered in the OTT formulation solved by Olkin and Pukelsheim
\cite{olkin1982distance}:
\[ \alpha = \Sigma^{\frac{1}{2}}_{X} \left( \Sigma^{\frac{1}{2}}_{X}
\Sigma_{X|Y} \Sigma^{\frac{1}{2}}_{X} \right)^{-\frac{1}{2}}
\Sigma^{\frac{1}{2}}_{X}. \]
As such, it follows that the optimal Brenier map is the $d$-dimensional
Schalkwijk-Kailath scheme \cite{schalkwijk1966csaOne}. This $d$-dimensional
scheme was also used to prove fundamental limits of control over noisy channels
in \cite{Elia04}.
\begin{remark}
  It can be easily verified that the 1-dimensional case of \cite[eqn. 22]{Shayevitz2011pm} that
  $X_{n+1} = \sqrt{1+\text{SNR}}(X_n - \frac{\text{SNR}}{1+\text{SNR}} Y_n)$ is a
  natural derivation of the results \eqref{eqn:snr} in higher
  dimension spaces.
\end{remark}

\subsubsection{The Optimal KR map and Successive Cancellation}
Now consider the parameterized family $(M_\epsilon)_{\epsilon > 0}$ of positive definite matrices given by
\eqref{eqn:KnotheRosenblatt:Mepsilon}.  Then these tend to the Knothe-Rosenblatt couplings, and analogous to Lemma~\ref{lemma:KnotheRosenblatt}, the optimal map in the two-dimensional Gaussian case is given by:
\beqa
X_{n+1}[1] &=& \beta_1 \parenth{X_n[1] - \E\brackets{X_n[1]|Y_n}} \nonumber\\
X_{n+1}[2] &=& \beta_2 \parenth{X_n[2]-\E\brackets{X_n[2]|Y_n}} +\beta_3 X_{n+1}[1] \nonumber
\eeqa
See  \cite{rosenblatt1952remarks} for the $\beta_1,\beta_2,\beta_3$ constants.
This can be naturally extended for arbitrary $d > 2$. Note that the essence of
this scheme is successive cancellation \cite{zhang2007successive} used to decode corner points in a multiple-access problem \cite{CoverThomas06}: first decode the first dimension of $W$,
followed by the second given knowledge of the first, etc.  Here, we are using
the chain rule to expand $I(W;Y) = I(W[1];Y) + I(W[2];Y|W[1])$.  Thus, this has
the potential to be applicable to unequal error protection scenarios where the
first dimension of $W$ has more important information than the second.

\section{Discussion and Conclusion} \label{sec:discussionConclusion}
In this paper, with the aid of optimal transport theory, we  generalized the
notion of posterior matching for message point feedback communication problems
from the $(0,1)$ interval to $d$-dimensional Euclidean space.  In addition, we
developed notions of reliability and achievability from an almost-sure
perspective and subsequently used classical probabilistic analysis methods to
establish \emph{succinct} necessary and sufficient conditions on when both
reliability and achieving capacity occur: Birkhoff ergodicity of the $(\tW_n)_{n
\geq 1}$ random process at the encoder.   The applications included
multi-antenna communications, and brain-computer interfaces.

Multiple theorems and lemmas established connections to modern applications
including intrinsic methods for stability of the nonlinear filter in hidden
Markov models and optimal transport theory.  It may be fruitful to further
cross-fertilize ideas between information theory and these areas.  For instance,
determining when the $(\tW_n)_{n \geq 1}$ random process is Birkhoff ergodic is
a natural next question.  In general, especially for continuous channels, this
may be very challenging.  Leveraging recent progress in optimal transport theory
may enable the design of certain cost functions which guarantee the optimal
$S_y^*$ results in ergodicity. Additionally, the surprising all-or-nothing
nature between reliably communicating under PM schemes and the Birkhoff
ergodicity of $(\tW_n)_{n \geq 1}$ could suggest new ways to transform questions
about the ergodicity of a random process, into questions about appropriate
channels under which they can be reliably communicated.

The use of optimal transport theory for posterior matching to construct a map to
transform a sample from the posterior distribution on the message point $W$ to
the prior distribution on $W$ is \emph{dual} to optimal transport methods
developed for Bayesian inference, where the objective is to construct a map that
transforms a sample from the prior distribution on the latent variable into a
sample from the posterior \cite{el2012bayesian}.   As such, this implies that
for the same pair of posterior and prior, solving for a map for one problem
simultaneously solves for the other (with taking an inverse).  Recently
developed convex optimization methods for Bayesian inference with optimal maps
\cite{kim2014dynamic,MesaKimColeman15ISIT,mesa2018NC} thus may possibly
leveraged for construction of maps within the context of posterior matching.

Future work could explore the Knothe-Rosenblatt optimal transport construction
and its possible connection to ``onion-peeling'' successive cancellation
decoding \cite{zhang2007successive}, as well as unequal error protection
\cite{nakiboglu2013bit}.  The fact that ergodicity is the if and only if
condition for message point recovery and optimal convergence rates, suggests
there might be further opportunities to use this type of mathematics to study
optimization over stochastic dynamical systems problems in areas such as team
decision theory (e.g. control over noisy channels
\cite{Elia04,ardestanizadeh2012control,kulkarni2015optimizer}) and statistical
physics \cite{mitter2005information}.

Additionally, future work could focus on further developing the use of this
framework for use in interacting systems of humans and machines beyond
traditional brain-computer interfaces. For example, in the field of interactive
reinforcement learning, some research has focus on systems in which a computer
attempts to learn an optimal control policy from a human where sequential
corrective action is being given by a human
\cite{knox2008tamer,vien2013learning,mericcli2010complementary}. This
generalization of the posterior matching scheme complements previous research
\cite{castro2009human,tsiligkaridis2014collaborative} into developing the
necessary theory to maximize the efficiency of such systems. 
\section*{Acknowledgments}
The authors thank Sean Meyn, Ramon Van Handel, Ofer Shayevitz, Meir Feder, Tara Javidi, Young-Han Kim, Maxim Raginsky, and Robert Gray for useful discussions.

\bibliographystyle{IEEEtran}
\bibliography{GMCPosteriorMatchingReferences}

\end{document}